\documentclass[pdflatex]{sn-jnl}% 

\usepackage{amsmath,amssymb,amsfonts}%

\usepackage{mathtools}
\usepackage{comment}
\usepackage{float}

\usepackage{multirow}
\usepackage{array}

\usepackage{graphicx}
\usepackage{subcaption}
\usepackage{amsthm}%
\usepackage{mathrsfs}%
\usepackage[title]{appendix}%
\usepackage{xcolor}%
\usepackage{textcomp}%
\usepackage{manyfoot}%
\usepackage{booktabs}%
\usepackage{algorithm}%
\usepackage{algorithmicx}%
\usepackage{algpseudocode}%
\usepackage{listings}
\theoremstyle{thmstyleone}
\newtheorem{theorem}{Theorem}
\newtheorem{proposition}{Proposition}

\newtheorem{corollary}{Corollary}
\newtheorem{example}{Example}
\newtheorem*{solution}{Solution}

\usepackage{bm}

\theoremstyle{thmstyletwo}%

\theoremstyle{thmstylethree}%
\newtheorem{definition}{Definition}%
\newtheorem{remark}{Remark}%

\usepackage{hyperref}

\begin{document}

\title[Extropy Rate: Properties and Application in Feature Selection]{Extropy Rate: Properties and Application in Feature Selection}

\author[1]{\fnm{Naveen} \sur{Kumar}}\email{kumar.248@iitj.ac.in}

\author[2]{\fnm{Vivek} \sur{Vijay}}\email{vivek@iitj.ac.in}


\affil[1,2]{\orgdiv{Department of Mathematics}, \orgname{Indian Institute of Technology Jodhpur}, \orgaddress{\street{Karwar}, \city{Jodhpur}, \postcode{342030}, \state{Rajasthan}, \country{India}}}

%\affil[2]{\orgdiv{Department}, \orgname{Organization}, \orgaddress{\street{Street}, \city{City}, \postcode{10587}, \state{State}, \country{Country}}}

%\affil[3]{\orgdiv{Department}, \orgname{Organization}, \orgaddress{\street{Street}, \city{City}, \postcode{610101}, \state{State}, \country{Country}}}

\abstract{Extropy, a complementary dual of entropy, (proposed by Lad et al. \cite{lad2015extropy} in 2015) has attracted considerable interest from the research community. In this study, we focus on discrete random variables and define conditional extropy, establishing key properties of joint and conditional extropy such as bounds, uncertainty reduction due to additional information, and Lipschitz continuity. We further introduce the concept of extropy rate for a stochastic process of discrete random variables as a measure of the average uncertainty per random variable within the process. It is observed that for infinite stationary and ergodic stochastic processes, as well as for identically and independently distributed sequences, the extropy rate exhibits asymptotic equivalence. We explore the extropy rate for finite stochastic processes and numerically illustrate its effectiveness in capturing the underlying information across various distributions, quantifying complexity in time series data, and characterizing chaotic dynamics in dynamical systems. The behaviour of estimated extropy rate is observed to be closely aligned with Simpson's diversity index. The real-life applicability of the extropy rate is presented through a novel feature selection method based on the fact that features with higher extropy rates contain greater inherent information. Using six publicly available datasets, we show the superiority of the proposed feature selection method over some other existing popular approaches.}

\keywords{Entropy rate, Extropy rate, Conditional extropy, Empirical distribution, Feature selection}

\msc{94A17, 62B10, 60G10, 62H30, 62P30}

\maketitle

\section{Introduction}\label{sec1}
Shannon entropy, a fundamental measure of uncertainty, information, randomness, and complexity, has contributed significantly to diverse fields such as artificial intelligence, information theory, statistics, thermodynamics, finance, and many more\cite{shannon1948mathematical,kumar2025entropy}. The core of quantifying uncertainty in a probability distribution lies in applying a concave function to the event probabilities, assigning greater weights to events with lower probability masses and achieving its maximum at the uniform distribution. Extropy\cite{lad2015extropy} emerges as a complementary dual measure, capturing uncertainty based on the probabilities of non-occurrence of events. Extropy has found real-world applications in pattern recognition\cite{balakrishnan2022tsallis,kazemi2021fractional,deng2024plausibility}, classification problems\cite{buono2024unified}, goodness-of-fit testing\cite{sathar2021dynamic,gupta2024some}, uniformity assessments\cite{qiu2018extropy}, system failure analysis\cite{qiu2019extropy}, and financial modelling\cite{tahmasebi2022negative}.

 In recent years, there has been growing interest in developing both theory and applications around the extropy. Notable contributions around the development of the varients of extropy include the cumulative residual extropy introduced by Jahanshahi et al.\cite{jahanshahi2020cumulative}, the negative cumulative extropy by Tahmasebi and Toomaj\cite{tahmasebi2022negative}, and the dynamic weighted extropy proposed by Sathar and Nair\cite{sathar2021dynamic}. Furthermore, Toomaj et al.\cite{toomaj2023extropy} extended this framework by formulating dynamic versions of extropy, while Saranya and Sunoj\cite{saranya2024relative} developed the relative cumulative extropy and its residual counterparts. Among the parametric generalizations, Balakrishnan et al.\cite{balakrishnan2022tsallis} proposed Tsallis extropy, Liu and Xiao\cite{liu2023renyi} introduced Rényi’s extropy, and Chakraborty and Pardhan\cite{chakraborty2023weighted} defined the cumulative Tsallis residual and past extropy. Building on Dempster–Shafer evidence theory, Buono and Longobardi\cite{buono2020dual} presented Deng extropy, followed by Kazemi et al.\cite{kazemi2021fractional}, who introduced the fractional Deng extropy. Buono et al.\cite{buono2024unified} also proposed a unified two-parameter extropy which contains several known forms as special cases, and Deng et al.\cite{deng2024plausibility} further defined the plausibility extropy. Most recently, Saha and Kayal\cite{saha2025copula} introduced the concept of copula extropy, which measures dependency structures among multiple variables.

Quantifying uncertainty in a stochastic process by calculating the uncertainty of each random variable separately fails to capture the dependency among the random variables inherent in the process. Entropy rate of a stochastic process is the average entropy per unit of time that reflects the average uncertainty associated with the random variables. It also represents the average number of bits required to describe the stochastic process over time fully. The entropy rate is used across various fields to track how information evolves in complex systems\cite{gomez2008entropy}. In networks, it helps explain how structure affects the spread of information\cite{latora1999kolmogorov}. In image analysis, it guides the creation of compact superpixels and highlights important regions by measuring information flow across image areas\cite{liu2011entropy,wang2010measuring}. It captures how quickly disorder grows in dynamical systems, including chaotic and near-chaotic cases, and links system behaviour to sensitivity and structure\cite{latora2000rate}. In physiological time series, it helps evaluate the regularity and complexity of signs like heartbeat, revealing changes over time\cite{porta2002entropy}. Some popular non-parametric estimates of entropy rate include approximate entropy\cite{pincus1991approximate} and sample entropy\cite{richman2004sample}. The study of the extropy rate offers a complementary perspective to the entropy rate by emphasizing the uncertainty associated with frequent events, in contrast to entropy, which assigns a higher weight to events with low probabilities. A significant challenge in defining extropy rate arises from the fact that extropy has different definitions for discrete and continuous random variables and lacks additivity, as it does not naturally split into the extropy of individual random variables, unlike entropy. This makes the formulation of the extropy rate for stochastic processes nontrivial. To the best of our knowledge, the concept of extropy rate has not yet been explored in the literature.

In this paper, we propose a definition of extropy rate for stochastic processes of discrete random variables based on a relation established by Lad et al. \cite{lad2015extropy} between entropy and extropy. We further investigate key properties of joint extropy for discrete random variables, an area still underdeveloped in literature. Numerical illustrations support the idea that the extropy rate can serve as a meaningful measure of information, complexity, bifurcation sensitivity in dynamical systems, and diversity.

The article is organized as follows. Section \ref{prel} presents some necessary preliminaries. Section \ref{extropyrate} introduces the extropy rate for stochastic processes of discrete random variables and presents key properties of extropy rate and joint extropy. Section \ref{numericalresults} demonstrates its performance as a measure of uncertainty through numerical results. Section \ref{applications} proposes a novel feature selection method based on extropy rate, highlighting its superiority over existing approaches. Finally, Section \ref{conclusion} concludes the paper.

\section{Preliminaries} \label{prel}

Let $X$ be a discrete random variable with probability distribution $\{p_i\}$. The Shannon entropy\cite{shannon1948mathematical} of $X$ is defined as 
\begin{equation}
    H(X)=-\sum_{i}p_i \log p_i.
\end{equation}
$H(X)$ is always non-negative, continuous in the probabilities $p_i$, attains its maximum at the uniform distribution, and is invariant under the $p_i$'s permutations. Let $\mathcal{X} = (X_1, X_2, \dots)$ be a stochastic process of discrete random variables, and let $\{p_{i_1, i_2, \dots, i_n}\}$ denotes the joint probability mass function of $(X_1, X_2, \dots, X_n)$. Then, the entropy rate of the process is defined as
\begin{equation}
    H_{rate}(\mathcal{X})=\lim_{n \to \infty}\frac{H(X_1, X_2, ..., X_n)}{n}.
\end{equation}
The entropy rate quantifies the minimum channel capacity required for error-free communication\cite{shannon1948mathematical}, measures the complexity of time-series data\cite{paluvs1996coarse}, indicates chaos in dynamic systems\cite{hilborn2000chaos}, and quantifies the information revealed per unit of time\cite{kennel2005estimating}. Lad et al.\cite{lad2015extropy} define extropy as a complementary dual of entropy, and for a discrete random variable $X$ with probability mass function $\{p_i\}_{i=1}^{n}$, it is calculated as
\begin{equation}
    J(X)=-\sum_{i=1}^{n}(1-p_i)\log(1-p_i).
\end{equation}
Like $H(X)$, the extropy $J(X)$ is also non-negative, continuous in the probabilities $p_i$, and attains its maximum at the uniform distribution. The property that makes extropy the complementary dual of entropy is the relation given by 
\begin{equation}
H(X)+J(X)=\sum_{i=1}^{n}H(\{p_i, 1-p_i\}), 
\end{equation}
which implies that for a given value of event-wise entropy or extropy and one of the two quantities (entropy or extropy), the other can be uniquely determined. Some applications of extropy and its inspired measures include quantifying uncertainty, testing data uniformity, pattern recognition, and measuring complexity in time-series data\cite{deng2024plausibility,qiu2018extropy,balakrishnan2022tsallis,giri2025permutation}.

\section{Extropy Rate}\label{extropyrate}

We now define the extropy rate for a stochastic process of discrete random variables and derive some of its fundamental properties. Note that the entropy of a finite discrete-time stochastic process $\mathcal{X} = \{X_1, X_2, \ldots, X_n\}$ increases linearly with the number of random variables in the process as 
\begin{equation}
H(\mathcal{X})=\sum_{j=1}^{n}H(X_j|X_{j-1}),  
\end{equation}
where 
\begin{equation}
    H(X|Y)=\sum_{y}p_y H(X|Y=y), 
\end{equation}
${p_y}$ is the pmf of the random variable $Y$. When the random variables are independent, this simplifies to 
\begin{equation}
H(\mathcal{X})=\sum_{j=1}^{n}H(X_j).  
\end{equation}
This linear growth allows to take the arithmetic mean of the total entropy across time to obtain a meaningful quantity known as the entropy rate, that is, the entropy rate 
\begin{equation}
H_{rate}(\mathcal{X})=\frac{H(\mathcal{X})}{n}
\end{equation}
reflects the average information produced per time step. In the case of an infinite discrete-time stochastic process $\mathcal{X}=\{X_1, X_2, ...\}$, the entropy rate is typically defined as the limiting average given by 
\begin{equation}
H_{rate}(\mathcal{X})=\lim_{n\to \infty}\frac{H(\mathcal{X})}{n}.
\end{equation}
More broadly, the definition of a rate for any function depends on how that function behaves when distributed across simpler system components. To develop a similar understanding of the extropy functional, we first define conditional extropy to capture the effect of one random variable on other random variables.

\begin{definition}[\textbf{Conditional Extropy}]
    Let $X$ and $Y$ be two random variables with joint probability distribution $p_{ij}$, the marginal distribution of $Y$ given by $p_j$, and the conditional distribution of $X$ given $Y = j$ denoted by $p_{i|j}$. The conditional extropy of $X$ given $Y = j$ is defined as
    \begin{equation}
     J(X \mid Y = j) = -\sum_{i}\left(1-p_{i|j} \right)\ln\left(1-p_{i|j}\right).
     \end{equation}
     The conditional extropy of $X$ given $Y$ is the weighted average defined as
     \begin{equation}\label{conditionalextropy}
     J(X \mid Y) = \sum_j p_j \, J(X \mid Y = j).
     \end{equation}
\end{definition}
However, joint extropy does not show the same additive behaviour as joint entropy, as shown in the following example.
\begin{example} Consider two discrete random variables \( X \) and \( Y \) with the following joint distribution:
\[
\begin{aligned}
P(X = 0, Y = 0) &= 0.1, \\
P(X = 0, Y = 1) &= 0.2, \\
P(X = 1, Y = 0) &= 0.3, \\
P(X = 1, Y = 1) &= 0.4.
\end{aligned}
\]

\noindent The conditional extropy \( J(Y \mid X) \), calculated using equation (\ref{conditionalextropy}), is given by
\[
J(X|Y) = 0.608993.
\]

\noindent The extropy of the random variable \( X \) is
\[
J(Y) = 0.673011.
\]

\noindent The joint extropy of $(X,Y)$ is
\[
J(X, Y) = 0.829507.
\]

\noindent Note that
\[
J(X, Y) \neq J(X|Y) + J(Y)
\]
illustrating that the relation is not generally valid for extropy.
\end{example}

The example shows that the value of joint extropy is not increasing linearly with the number of random variables in the stochastic process. In the following proposition, we examine how this behaviour affects the definition of extropy rate if we define extropy rate as the average of the joint extropy over a sequence of random variables, that is,
\begin{equation}\label{defination1}
    R(\mathcal{X})=\lim_{n\to \infty}\frac{J(X_1,...,X_n)}{n}.
\end{equation}
\begin{proposition}\label{prepositionofzero}
    Let $\bm{\mathcal{X}}=\{X_1, X_2, \ldots\}$ be a discrete-time stochastic process with a finite state space and the support of $X_i$ is \( m_i(>1) \) for $i=1,2,...$, then $R\left(\bm{\mathcal{X}}\right)=0$.
\end{proposition}
\begin{proof}
    Note that $R\left(\bm{\mathcal{X}}\right)$ is a non-negative function and we know that extropy attains its maximum value under the uniform distribution\cite{lad2015extropy}. Now 
    \begin{equation}
    \begin{split}
    R\left(\bm{\mathcal{X}}\right)&=\lim_{n\to \infty}\frac{J(X_1, X_2,...,X_n)}{n} \\
    &=\lim_{n\to \infty}\frac{-\sum_{i}\left(1-p_{i_1,i_2,...,i_n}\right)\ln\left(1-p_{i_1,i_2,...,i_n}\right)}{n} \\
    &\leq \lim_{n\to \infty}\frac{-\sum_{i}\left(1-\frac{1}{\prod_{i=1}^{n}m_i}\right)\ln\left(1-\frac{1}{\prod_{i=1}^{n}m_i}\right)}{n} \\
    &=\lim_{n\to \infty} \frac{1}{n} \frac{\ln\left(1-\frac{1}{\prod_{i=1}^{n}m_i}\right)}{\left(\frac{1}{\prod_{i=1}^{n}m_i-1}\right)}=0.
    \end{split}
    \end{equation}
    It follows from the fact that 
    \begin{equation}
        \lim_{x\to \infty}\frac{\log\left(1-\frac{1}{x}\right)}{\frac{1}{x-1}}=-1.
    \end{equation}
    This implies that $R\left(\bm{\mathcal{X}}\right)=0$.
\end{proof}

This shows that the extropy rate, as defined in Equation (\ref{defination1}), tends to zero for any stochastic process of discrete random variables. As a result, it fails to provide meaningful inferences around the underlying physical phenomena. In contrast, the entropy rate captures the degree of dependence among random variables by quantifying the uncertainty associated with each variable, which can change (increase or decrease) when additional variables are introduced, reflecting both uncertainty and interdependence. Lad et al. \cite{lad2015extropy} showed that for a probability distribution $\{p_1, \ldots, p_m\}$, the entropy corresponds to the distribution $\{q_1, \ldots, q_m\}$ where each $q_i = \frac{1 - p_i}{m - 1}$, is a rescaled value of extropy of the distribution $\{p_1, \ldots, p_m\}$ given by
\begin{equation}
H\left(\{q_1,q_2,...,q_m\}\right)=\log(m-1)+\frac{J\left(\{p_1,p_2,...,p_m\}\right)}{m-1}.
\end{equation}
This rescaled version establishes extropy as a complementary measure to entropy and also serves as the inspiration for defining the extropy functional of probability distributions by removing the rescaling component. Motivated by this interpretation, we propose the following definition of the extropy rate.
\begin{definition}[\textbf{Extropy Rate}]

Let $\bm{\mathcal{X}} = \{X_1, X_2, \dots\}$ be an infinite discrete-time discrete-state stochastic process. The extropy rate of $\bm{\mathcal{X}}$ is defined by 
\begin{equation}
    J_{rate}\left(\bm{\mathcal{X}}\right)=\lim_{n\to \infty}\left(\frac{1}{n} \right)\left(\log({S_n}-1)+\frac{J(X_1, X_2,...,X_n)}{{S_n}-1}\right),
\end{equation}
where ${S_n}$ is the support size of $(X_1,X_2,...,X_n)$ for all $n$.
\end{definition}

% Following the definition of extropy, it is evident that Lad et al. \cite{lad2015extropy} defined the extropy of a discrete random variable only for cases with finite support. For this, one key reason is that for discrete random variables with infinite support, such as geometric, Poisson, or Zipf distributions, the extropy becomes zero, as shown in Proposition [1]. This makes extropy uninformative or uninteresting for discrete random variables with infinite support sizes. Moreover, unlike Shannon entropy, the functional forms of extropy differ between the discrete and continuous cases. Due to these limitations, we focus on defining the extropy rate for finite discrete-time discrete-state stochastic processes, where each random variable has finite support. Extending the analysis to infinite discrete-time processes involving non-trivial random variables would cause the effective support size to diverge, leading the extropy to converge to zero.

%\subsection{Shannon approach}

Extropy measures the uncertainty inherent in a random variable $X_i$. Thus, the extropy rate of a stochastic process measures how much uncertainty or new information is produced on average at each time step. It reflects the average amount of unpredictable in an outcome based on the history of the process $\bm{\mathcal{X}}$.

\begin{remark}
The entropy $H$ of a discrete random variable is always non-negative by definition, implying that the extropy rate $J_{rate}(\mathcal{X})$ also remains non-negative.
\end{remark}

\begin{remark}
    Note that the extropy of any discrete distribution is bounded above by $1$. Therefore, we have 
     \begin{equation}
         \lim_{n\to \infty}\frac{J(X_1,X_2,...,X_n)}{n}=0.
     \end{equation}
    Further, we can rewrite the defination of extropy rate for a infinite stochastic process $\mathcal{X}=\{X_i\}_{i}$ as
    \begin{equation}\label{infiniteperiodextropyrate}
        J_{rate}(\mathcal{X})=\lim_{n\to \infty}\frac{\log(S_n-1)}{n}.
    \end{equation}
    This expression approximately equal to the average of zeroth-order Rényi entropy\cite{renyi1961measures}, given by
    \begin{equation}
        R_{\alpha}(X)=\frac{1}{\alpha-1}\log \left(\sum_{i=1}^{\infty}p_i^{\alpha} \right), \textit{ \ $0<\alpha<\infty,$ \ $\alpha\neq 1$.}
    \end{equation}
\end{remark}
The following remark establishes a connection between the extropy rate and topological entropy.
\begin{remark}
    Let $X$ be a compact topological space and $f: X \to X$ a continuous map. Let $\mathcal{U}$ be a finite open cover of $X$, and let $N(\mathcal{U})$ denote the minimum number of sets from $\mathcal{U}$ required to cover $X$. The entropy of the cover $\mathcal{U}$ is defined as $H(\mathcal{U}) = \log N(\mathcal{U})$. For two open covers $\mathcal{A}$ and $\mathcal{B}$, their joint $\mathcal{A} \vee \mathcal{B}$ denotes the common refinement consisting of all non-empty intersections $A \cap B$, with $A \in \mathcal{A}$, $B \in \mathcal{B}$. The topological entropy\cite{adler1965topological} of $f$ with respect to $\mathcal{U}$ is then given by
    \begin{equation}
H(f, \mathcal{U}) = \lim_{n \to \infty} \frac{1}{n} H\left( \mathcal{U} \vee f^{-1}\mathcal{U} \vee \cdots \vee f^{-n+1}\mathcal{U} \right). 
\end{equation}
This quantity captures the average growth rate of distinguishable orbits under $f$, quantifying the long-term complexity of trajectories. The extropy rate, defined as the time-averaged logarithmic measure of cumulative sample proportions minus one for infinite time-period given by Equation (\ref{infiniteperiodextropyrate}), exhibits behaviour analogous to topological entropy in dynamical systems. For processes observed over an infinite time horizon, both topological entropy and extropy rate serve as asymptotic measures of system complexity, reflecting the average logarithmic growth of information or distinguishability over time.
\end{remark}

In real-world scenarios, many systems operate over only a finite period, represented by sequence of observations, say $\mathcal{X}=\{X_1,X_2 \ldots, X_n\}$. To accurately determine the extropy per component or random variable, it becomes necessary to estimate the extropy rate using 
\begin{equation}
    J_{rate}^{F}(\mathcal{X})=\left(\frac{1}{n} \right) \left(\log\left(S_n-1 \right)+\frac{J(X_1,X_2,...,X_n)}{S_n-1} \right).
\end{equation}

The next theorem presents a result for a sequence of independently and identically distributed (IID) random variables.
\begin{theorem}
Let \(\mathcal{X}= \{X_i\}_{i=1}^\infty \) be an IID sequence of discrete random variables with support size \( k \) then the extropy rate $J_{rate}(\mathcal{X})$, exists and equals to \( \log k \) approximately.
\end{theorem}
\begin{proof}
    From preposition (\ref{prepositionofzero}), we know that 
    \begin{equation*}
        \lim_{n\to \infty}\frac{J(X_1,X_2,...,Xn)}{n}=0
    \end{equation*}
    and the random variables in the process are IID, this implies that $S_n=k^n$. Thus from definition of extropy rate, we have
    \begin{equation}
    \begin{split}
        \lim_{n\to \infty}\frac{1}{n} \left(\log({S_n}-1)+\frac{J(X_1, X_2,...,X_n)}{{S_n}-1}\right)&=\lim_{n\to \infty}\frac{\log({S_n}-1)}{n}+\lim_{n\to \infty}\frac{J(X_1, X_2,...,X_n)}{n({S_n}-1)} \\
        &= \lim_{n\to \infty}\frac{\log({k^n}-1)}{n}\approx \log(k).
    \end{split}
    \end{equation}
    This completes the result.
\end{proof}

One way to interpret the definition of extropy is as the entropy evaluated over the total probability mass associated with the non-occurrence of each event. Notably, the expression for extropy of a random variable $X$ can also be rewritten as
\begin{equation}\label{generalizedextropy}
    J\left( X\right)=-\sum_{i=1}^{n}\left(1-p_i\right)\ln\left(1-p_i\right)=-\sum_{i=1}^{n}\left(\sum_{j=1}^{n}p_j-p_i\right)\ln\left(\sum_{j=1}^{n}p_j-p_i\right).
\end{equation}
It is easy to see that the definition of extropy is based on the total probability of the non-occurrence of events. However, there are scenarios where subtracting the probability of an event's occurrence from one does not necessarily give the actual probability of its non-occurrence. For example, in a judicial decision where the outcomes may be `innocent', `guilty', or `indeterminable' due to insufficient evidence, the complement of the probability of one outcome does not directly represent the probability of its non-occurrence. In such contexts, the interpretation of non-occurrence is nuanced and context-dependent. Therefore, we refer to the definition provided in equation (\ref{generalizedextropy}) as the generalized extropy and denoted by $\prescript{\{p_i\}}{}J\left(X\right)$. The following proposition establishes the result for conditional extropy in the case of two independent random variables.
\begin{proposition}
    If $X$ and $Y$ are independent random variables, then
$J(X \mid Y) = J(X)$.
\end{proposition}

\begin{proof}
    Since $X$ and $Y$ are independent random variables so $p_{i|j}=p_i$, for all $i,j$. Consider 
    \begin{equation}
    \begin{split}
        J(X|Y)&=-\sum_{j}p_{j}J(X|Y=j) \\
             &=-\sum_{i,j}p_{j}\left(1-p_{i|j}\right)\ln\left(1-p_{i|j}\right) \\
             &= -\sum_{i,j}p_j\left(1-p_i \right)\ln\left(1-p_i \right) \\
             &= -\sum_{i}\left(1-p_i \right)\ln\left(1-p_i \right) \\
             &=J(X).
    \end{split}
    \end{equation}
    This completes the proof.
\end{proof}
\begin{comment}
\begin{proposition}
    Let $n$ be a positive integer, and let $X_1, X_2, \ldots, X_n$ be $n$ finite-state random variables. Then, the joint extropy satisfies
    \begin{equation}
    J(X_1, X_2, \ldots, X_n) \geq J(X_1, X_2, \ldots, X_{n-1} \mid X_n).
    \end{equation}
\end{proposition}
\begin{proof}
    From definition, we have
    \begin{equation}
    \begin{split}
        J(X_1, X_2, \ldots, X_{n-1} \mid X_n)&=-\sum_{i_1, i_2, \ldots, i_n} p_{i_n} \left(1-p_{i_1, i_2, \ldots, i_{n-1}\mid i_n} \right) \ln \left(1-p_{i_1, i_2, \ldots, i_{n-1}\mid i_n}\right) \\
        &= -\sum_{i_1, i_2, \ldots, i_n}\left(p_{i_n}-p_{i_1, i_2, \ldots, i_{n-1}, i_n} \right) \ln \left(1-p_{i_1, i_2, \ldots, i_{n-1}\mid i_n}\right) \\
        &\geq -\sum_{i_1, i_2, \ldots, i_n}\left(p_{i_n}-p_{i_1, i_2, \ldots, i_{n-1}, i_n} \right) \ln \left(1-p_{i_1, i_2, \ldots, i_{n-1}, i_n}\right) \\
        &\geq -\sum_{i_1, i_2, \ldots, i_n}\left(1-p_{i_1, i_2, \ldots, i_{n-1}, i_n} \right) \ln \left(1-p_{i_1, i_2, \ldots, i_{n-1}, i_n}\right) \\
        &=J(X_1, X_2, \ldots, X_n)
    \end{split}
    \end{equation}
\end{proof}
\end{comment}
\noindent We now derive a bound on the joint extropy in terms of the individual extropy of the random variables.
\begin{proposition}
    Let $p_{ij}$, $p_i$, $p_j$, and $p_{i|j}$ denote the joint pmf of $(X, Y)$, the marginal pmf of $X$ and $Y$, and the conditional pmf of $X$ given $Y$, respectively. Then, for $\max_{i,j}p_{ij} < 1-\frac{1}{e}$,
    \begin{equation}
        J\left( X, Y \right) \leq \min \{m_YJ(X), m_XJ(Y), J(X|Y) \}
    \end{equation}
    and for $\min_{i,j}\{p_i, p_j, p_{i|j}\}>1-\frac{1}{e}$, 
    \begin{equation}
        J\left( X, Y \right) \geq \max \{m_YJ(X), m_XJ(Y), J(X|Y) \},
    \end{equation}
    where $m_X$ and $m_Y$ are the size of support of $X$ and $Y$.
\end{proposition}
\begin{proof}
    Consider 
    \begin{equation}
        \begin{split}
            \frac{d}{dx}\left[-\left(1-x \right) \ln\left(1-x \right)\right]&=0 \\
            1+\ln\left(1-x \right)&=0.
        \end{split}
    \end{equation}
Which implies that $-\left(1-x \right) \ln\left(1-x \right)$ is a increasing function in $\left(0, 1-\frac{1}{e}\right)$ and decreasing in $\left(1-\frac{1}{e},1\right)$. Also, 
\begin{equation}
p_{ij}=p_{i|j}p_j\leq p_{i|j}/ p_{j},
\end{equation}
and
\begin{equation}
p_{ij}=p_{j|i}p_i\leq p_{j|i}/ p_{i}.
\end{equation}
Thus, if $\max_{ij}p_{ij}\leq1-\frac{1}{e}$, we have 
\begin{equation}
    -\left(1-p_{ij} \right)\ln\left(1-p_{ij} \right)\leq -\left(1-p_{i} \right)\ln\left(1-p_{i} \right),
\end{equation}
and 
for $\min_{i}\{p_i\}\geq1-\frac{1}{e}$, we have
\begin{equation}
    -\left(1-p_{ij} \right)\ln\left(1-p_{ij} \right)\leq -\left(1-p_{i} \right)\ln\left(1-p_{i} \right).
\end{equation}
Summing both sides for all $i$ and $j$ gives
\begin{equation}
    \sum_{i,j}-\left(1-p_{ij} \right)\ln\left(1-p_{ij} \right)\leq\sum_{i,j} -\left(1-p_{i} \right)\ln\left(1-p_{i} \right)
\end{equation}
and 
\begin{equation}
    \sum_{i,j}-\left(1-p_{ij} \right)\ln\left(1-p_{ij} \right)\geq\sum_{i,j} -\left(1-p_{i} \right)\ln\left(1-p_{i} \right)
\end{equation}
for $\max_{ij}p_{ij}<1-\frac{1}{e}$ and $\min_{i}\{p_i\}>1-\frac{1}{e}$, respectively. Thus, for $\max_{ij}p_{ij}<1-\frac{1}{e}$,
\begin{equation}
    J\left(X,Y \right)\leq m_Y J\left(X \right)
\end{equation}
and for $\min_{i}\{p_i\}> 1-\frac{1}{e}$, we get
\begin{equation}
    J\left(X,Y \right)\geq m_Y J\left(X \right).
\end{equation}
Similarly, we can prove the inequalities for $J(Y)$ and $J(X|Y)$. This completes the proof.
\end{proof}
The following result expresses conditional extropy in terms of generalized conditional extropy and the entropy of the constituent random variables.
% \begin{remark}
%     Let $\bm{\mathcal{X}} =\{X_1, X_2, \ldots\}$ be a discrete-time finite-state stochastic process consisting of independent and identically distributed random variables, and let $\{p_k\}$ be the probability mass function of $X_k$ for all $k = 1, 2, \ldots$. Then, for $\max_{i_1, i_2,..., i_n}p_{i_1, i_2, ...,i_n}<1-\frac{1}{e}$, we have \begin{equation}
%         J_{rate}^{D}(\bm{\mathcal{X}})\leq\lim_{n\to \infty}\frac{m_{X_1} J^D(X_1)}{n}
%     \end{equation}
%     and for $\min_i p_i>1-\frac{1}{e}$, we have \begin{equation}
%         J_{rate}^{D}(\bm{\mathcal{X}})\geq\lim_{n\to \infty}\frac{m_{X_1} J^D(X_1)}{n},
%     \end{equation}
%     where $\{p_{i_1, i_2, ...,i_n}\}$ is the probability distribution of $(X_1,X_2,...,X_n)$.
% \end{remark}

\begin{proposition}
    $J\left( X|Y \right)=\prescript{\{p_j\}}{}J\left(X|Y \right) + (m_{X}-1)H\left(Y\right)$.
\end{proposition}

\begin{proof}
    Consider
    \begin{equation}
    \begin{split}
        J\left( X|Y \right)&=-\sum_{i,j}p_j\left(1-p_{i|j} \right)\ln\left(1-p_{i|j} \right) \\
        &= -\sum_{i,j}\left(p_j-p_{i,j} \right)\ln\left(\frac{p_j-p_{i,j}}{p_j} \right) \\
        &= -\sum_{i,j}\left(p_j-p_{i,j} \right)\left(\ln\left({p_j-p_{i,j}}\right) -\ln \left(p_j\right) \right)\\
        &=-\sum_{i,j}\left(p_j-p_{i,j} \right)\ln\left({p_j-p_{i,j}}\right) + \sum_{i,j}\left(p_j-p_{i,j} \right)\ln\left(p_j\right) \\
        &=\prescript{\{p_j\}}{}J\left(X|Y \right) + (m_{X}-1)H\left(Y\right).
    \end{split}
    \end{equation}
    %Here $J^{p_j}$ is the generalized extropy associated with the distribution $\{p_j\}$.
\end{proof}
The following proposition establishes that the extropy associated with a phenomenon decreases when additional information is considered.
\begin{proposition}
   $J\left( Y\right)> J\left( Y|X\right)$.
\end{proposition}
\begin{proof}
    Considering a function $f(x)=-(1-x)\ln(1-x)$, it is easy to see that $f$ is a strictly concave function on $(0,1)$ as $f^{''}(x)<0$. Since $p_{ij}=p_{i|j}p_j$ so from the definition of concave function for all $j$, we have 
    \begin{equation}
    \begin{split}
        f\left(p_j\right)&=f\left(\sum_{i}p_{i,j}\right)   \\
              &=f\left(\sum_{i}p_{j|i}p_i\right)  \\
              &> \sum_{i}p_i f\left(\sum_{i}p_{j|i}\right) \\
              &=-\sum_{i}p_i \left(1-p_{j|i}\right) \ln\left(1-p_{j|i}\right)
    \end{split}
    \end{equation}
    Taking the summation of $j$ on both sides, the result is established.
\end{proof}

% \begin{proposition}
% Let $\bm{\mathcal{X}}=\{X_1, X_2, \ldots\}$ be an independently distributed discrete-time stochastic process, then $J_{rate}^{D}\left(\bm{\mathcal{X}}\right)=0$.
% \end{proposition}

% % \\
% %         &=-\frac{1}{n}\sum_{i_1, i_2, ..., i_n}\left(1-p_{i_1}^n\right) \ln\left(1-p_{i_1}^n\right) \\
% %         &=-\frac{1}{n}\sum_{i_1}\left(1-p_{i_1}^n\right) \ln\left(1-p_{i_1}^n\right)

% \vspace{-0.7cm}
% \begin{proof}
%     Let $m$ be the support of $X_i$ for each $i=1,2,...$. Consider 
%     \begin{equation}
%         -\frac{1}{n}\sum_{i_1, i_2, ..., i_n}\left(1-p_{i_1, i_2, ..., i_n}\right) \ln\left(1-p_{i_1, i_2, ..., i_n}\right) = -\frac{1}{n}\sum_{i_1, i_2, ..., i_n}\left(1-p_{i_1} p_{i_2}\cdots p_{i_n}\right) \ln\left(1-p_{i_1} p_{i_2}\cdots p_{i_n}\right),
%     \end{equation}
%     which follows from the independence of all $X_i$'s. Now, for $0<p_{i_j}< 1$ where $j=1,2,...,n$, we have 
%     \begin{equation}
%     \begin{split}
%         \lim_{n\to \infty}p_{i_1} p_{i_2}\cdots p_{i_n} &\leq \lim_{n\to \infty} {\left(\max_{j=1,2,...,n}p_{i_j}\right)}^n\\
%         &= \lim_{n\to \infty} {\left(1-\epsilon \right)}^{n} \textit{ \ \ \ \ \ for some $\epsilon>0$, }\\
%         &=0.
%     \end{split}
%     \end{equation}
%     This implies that
%     \begin{equation}
%         \lim_{n\to \infty}-\frac{1}{n}\sum_{i_1, i_2, ..., i_n}\left(1-p_{i_1, i_2, ..., i_n}\right) \ln\left(1-p_{i_1, i_2, ..., i_n}\right)=0.
%     \end{equation}
%     This proves the result.
% \end{proof}

Let $E_n = \mathcal{X}_1 \times \mathcal{X}_2 \times \cdots \times \mathcal{X}_n$ denote the finite product space of support of the discrete random variables \( X_1, X_2, \dots, X_n \), where each $\mathcal{X}_i$ is finite. Define the set \( \mathcal{P}^r_n \) of all joint probability distributions over \( E \) as
\[
\mathcal{P}^r_n = \left\{ P = (p_{\mathbf{e}})_{\mathbf{e} \in E} \;\middle|\; 0 \leq p_{\mathbf{e}} < r \text{ for all } \mathbf{e} \in E,\; \sum_{\mathbf{e} \in E} p_{\mathbf{e}} = 1 \right\},
\]
where each \( p_{\mathbf{e}} = \mathbb{P}(X_1 = e_1, X_2 = e_2, \dots, X_n = e_n) \) for \( \mathbf{e} = (e_1, e_2, \dots, e_n) \in E \) and $r\in [0,1)$. Let $\mathcal{P}^r=\bigcup_{n=1}^{\infty}\mathcal{P}_n^r$. For a vector \( \bm{{v}} = (v_1, v_2, \ldots) \), the \( \ell_1 \)-norm is defined as
\[
\|\bm{{v}}\|_1 = \sum_{i=1}^{\infty} |v_i|.
\]
Under these settings, we establish the following analytic property of the extropy.

\begin{theorem}\label{theoremlipschitz}
For a given \( r \in [0,1) \), the extropy functional \( J \) is Lipschitz continuous on the set \( \mathcal{P}^r_n \) with Lipschitz constant \( \zeta \), under the metric induced by the \( \ell_1 \)-norm. That is, for any \( P, Q \in \mathcal{P}^r_n \),
\[
|J(P) - J(Q)| \leq \zeta \|P - Q\|_1,
\]
where $\zeta=\max_{P\in \mathcal{P}^r_n}\max_{P}\left({1+\bigg|\frac{1}{{\left(1-r\right)}^2}-1-2r\bigg|}\right)$.
\end{theorem}

% \begin{theorem}
% For a given $n\in \mathbb{N}$ and $r\in [0,1)$, the functional \( J_{rate}^{D} \) is Lipschitz continuous on \( \mathcal{P}^r_n \) with Lipschitz constant \( \zeta \) under the metric induced by \( \ell_1 \)-norm. That is, for any \( P, Q \in \mathcal{P}^r_n \),
% \[
% |J_{rate}^{D}(P) - J_{rate}^{D}(Q)| \leq \zeta \|P - Q\|_1.
% \]
% \end{theorem}
% \vspace{-0.7cm}
\begin{proof}
    Define the function \( f : [0, 1) \to \mathbb{R} \) by $f(x) = -(1 - x)\ln(1 - x)$. For $x,y\in [0,1)$, consider
    \begin{equation}
    \begin{split}
        |f(x)-f(y)|&=|-(1 - x)\ln(1 - x)+(1 - y)\ln(1 - y)| \\
                   &=\Bigg|(1-x)\left(x+\frac{x^2}{2}+\frac{x^3}{3}+... \right)-(1-y)\left(y+\frac{y^2}{2}+\frac{y^3}{3}+... \right)\Bigg| \\
                   &=\Bigg| \left(x-\frac{x^2}{2}-\frac{x^3}{6}-...\right)-\left(y-\frac{y^2}{2}-\frac{y^3}{6}-...\right) \Bigg|\\
                   &=\Bigg| (x-y)-\frac{(x^2-y^2)}{2}-\frac{(x^3-y^3)}{6}-... \Bigg| \\
                   &=\Bigg| (x-y)\left( 1-\frac{(x+y)}{2}-\frac{(x^2+xy+y^2)}{6}-... \right) \Bigg| \\  & \ \ \ \ \ \ \ \ \ \ \ \ \ \ \ \ \ \ \ \left(\textit{ as for all $n\in \mathbb{N}$, $x^n-y^n=(x-y)\sum_{j=0}^{n-1}x^{n-j-1}y^j$ }\right) \\
                   &\leq |x-y|\Bigg| 1-\frac{(x+y)}{2}-\frac{(x^2+xy+y^2)}{6}-... \Bigg|=|x-y|\Bigg|1-\sum_{n=2}^{\infty}\sum_{j=0}^{n-1}x^{n-j-1}y^j \Bigg|.
    \end{split}
    \end{equation}
    It is easy to see that the function $f$ is Lipschitz continuous if the function $|1-\sum_{n=2}^{\infty}\sum_{j=0}^{n-1}x^{n-j-1}y^j |$ is bounded above. Let us define the function
\begin{equation}
s_n(x, y) = x^n + x^{n-1}y + x^{n-2}y^2 + \cdots + y^n
\end{equation}
for $n \in \mathbb{N} \setminus \{1\}$. This is a polynomial of degree \( n \), and it contains \( n+1 \) monomials. Also, we have
\begin{equation}
|s_n(x, y)| \leq \sum_{j=0}^{n} |x|^{n-j} |y|^j.
\end{equation}
Let \( r = \max\{ |x|, |y| \} \), then
\begin{equation}
|s_n(x, y)| \leq \sum_{j=0}^{n} r^{n-j} r^j = \sum_{j=0}^{n} r^n = (n+1)r^n.
\end{equation}
Therefore
\begin{equation}
|s_n(x, y)| \leq (n+1)r^n.
\end{equation}
Now, from the result \cite{mathtex}, we have 
\begin{equation}
\begin{split}
        \bigg|\sum_{n=2}^{\infty}s_n(x, y)\bigg|&\leq \bigg|\sum_{n=2}^{\infty}(n+1)r^n\bigg|\\
        &=\bigg|\sum_{n=0}^{\infty}(n+1)r^n-1-2r\bigg| \\
        &=\bigg|\frac{1}{{\left(1-r\right)}^2}-1-2r\bigg| \textit{ \ \ for $0\leq r<1$.}
\end{split}
\end{equation}
This gives that 
\begin{equation}\label{rfunction}
\begin{split}
   \Bigg|1-\sum_{n=2}^{\infty}\sum_{j=0}^{n-1}x^{n-j-1}y^j \Bigg|&\leq 1+\Bigg|\sum_{n=2}^{\infty}\sum_{j=0}^{n-1}x^{n-j-1}y^j\Bigg| \\ &\leq 1+\bigg|\frac{1}{{\left(1-r\right)}^2}-1-2r\bigg|=\beta_{r(x,y)} \textit{(say)} \textit{ \ \ \ \ \ \ \ for $|r|<1$.} 
\end{split}
\end{equation}
This implies that 
\begin{equation}
    |f(x)-f(y)|\leq \beta_{r(x,y)} |x-y|, \textit{ \ \ for all $x,y\in [0,1)$}.
\end{equation}
Thus, the function $f$ is Lipschitz continuous with Lipschitz constant $\beta_{r(x,y)}$. Moreover, it is known that if $f_1, f_2, \ldots, f_n$ are Lipschitz continuous functions with respective Lipschitz constants $b_1, b_2, \ldots, b_n$ then the sum $f_1 + f_2 + \cdots + f_n$ is also Lipschitz continuous with Lipschitz constant $b = \sup\{b_1, b_2, \ldots, b_n\}$. Let \( P = \{p_i\}_{i=1}^{n_1}, \quad Q = \{q_j\}_{j=1}^{n_2} \in \mathcal{P}^r_n \)
 and without loss of generality, assume that $p_i=0$, for $i=n_1+1, n_1+2,...,n_2$ if $n_1<n_2$, then we have
\begin{equation}
\begin{split}
    |J(P)-J(Q)|&=\bigg|\sum_{i}^{n_2}f(p_i)-\sum_{i=1}^{n_2}f(q_i)\bigg|\\
    &\leq \zeta \|P - Q\|_1
\end{split}
\end{equation}
where $\zeta=\max_{P,Q\in \mathcal{P}^r_n}\max_{p,q\in P\cup Q}{\beta_{r(p,q)}}$. This proves the result.
\end{proof}

\begin{remark}
The above result does not hold for $r=1$, as the maximum value of ${\beta_{r(p,q)}}$ may not be finite, see equation (\ref{rfunction}). However, \( J \) is uniformly continuous for each \(r\in [0,1] \) such that for any \( \varepsilon > 0 \), there exists \( \delta = \varepsilon / \zeta > 0 \) such that for any probability distribution \( P, Q \in \mathcal{P}^1\) with \(\zeta \|P - Q\|_1 < \delta \), we have \( |J(P)-J(Q)| < \varepsilon \).
\end{remark}
\noindent Let $\mathcal{P}_n^{r,S_n}\subseteq \mathcal{P}_n^r$ be defined by 
\begin{equation}
    \mathcal{P}_n^{r,S_n} = \left\{ P = (p_{\mathbf{e}})_{\mathbf{e} \in E} \;\middle|\; 0 \le p_{\mathbf{e}} < r,\; \sum_{\mathbf{e} \in E} p_{\mathbf{e}} = 1,\; \left| \left\{ \mathbf{e} \in E : p_{\mathbf{e}} 
    > 0 \right\} \right| = S_n \right\},
\end{equation}
then we have the following corollary on the extropy rate.
\begin{corollary}
    Let $\mathcal{X} = (X_1, X_2, \ldots, X_n)$ be a finite stochastic process, and let $P, Q \in \mathcal{P}_{n}^{r,S_n}$. Then, for any $r\in (0,1)$, the functional extropy $J_{rate}^F(\mathcal{X})$ is Lipschitz continuous with Lipschitz constant $\zeta^{*}=\frac{\zeta}{n(S_n-1)}$.
\end{corollary}
\begin{proof}
    For any $P,Q \in P_{n}^{r,S_n}$, consider 
    \begin{equation}
    \begin{split}
        |J_{rate}^F(P)-J_{rate}^F(Q)| &= \left(\frac{1}{n(S_n-1)} \right)(J(P)-J(Q)) \\
        & < \zeta^{*}\|P - Q\|_1
    \end{split}
    \end{equation}
    where $\zeta^{*}=\zeta/{\left( n(S_n-1) \right)}$ from theorem \ref{theoremlipschitz}. This completes the proof.
\end{proof}

\noindent The chain rule for Shannon entropy of two random variables \( X \) and \( Y \) is given by
\[
H(X, Y) = H(X \mid Y) + H(Y).
\]

However, the following result establishes that extropy exhibits subadditive behaviour.

\begin{theorem}
Let \( X \) and \( Y \) be two independent random variables. Then the extropy of the joint distribution satisfies the subadditivity property:
\[
J(X, Y) < J(X) + J(Y).
\]
\end{theorem}
\begin{proof}
    Let $f:[0,1] \to \mathbb{R}$ be a function defined by $f(z)=(1-z)\ln(1-z)+z$, then
    \begin{equation}
        \frac{d}{dz}f(z)=-\ln(1-z)>0 \textit{ \ \ for all $z \in [0,1]$.}
    \end{equation}
This implies that $f$ is an increasing function on $(0,1)$. Also, $f(0)=0$ and $f(1)=1$ gives that for all $z\in [0,1]$, $f(z)\geq 0$.
Thus, we can say that for all $z\in [0,1]$, 
\begin{equation}\label{subeq}
    (1-z)\ln(1-z)\geq -z.
\end{equation}
Consider
\begin{equation}
\begin{split}
    J(X)+J(Y)-&J(X,Y) \\ &=\sum_{i,j}\left(1-p_{i,j}\right)\ln\left(1-p_{i,j}\right)-p_j\left(1-p_{i}\right)\ln\left(1-p_{i}\right)-p_i\left(1-p_{j}\right)\ln\left(1-p_{j}\right) \\
    &=\sum_{i,j}\left(1-p_{i}p_j\right)\ln\left(1-p_{i}p_j\right)-p_j\left(1-p_{i}\right)\ln\left(1-p_{i}\right)-p_i\left(1-p_{j}\right)\ln\left(1-p_{j}\right) \\
    &\geq \sum_{i,j} -p_i p_j+p_jp_i +p_ip_j>0 
\end{split}
\end{equation}
The inequality follows from equation (\ref{subeq}), which establishes the subadditivity.
\end{proof}
\begin{remark}
For a discrete-time finite IID stochastic process \(\mathcal{X}= \{X_t\}_{t=1}^n \), the extropy of the joint distribution satisfies
\[
J(X_1, X_2, \ldots, X_n) < n J(X_1).
\]
\end{remark}

Let \( X_1, X_2, \ldots \) be a sequence of IID\ discrete random variables following the distribution \( P \), and let \( \mathcal{S} \) denote the set of values that each \( X_i \) can take. For any natural number \( n \), define the \( n \)-length random vector \( X^n = (X_1, X_2, \ldots, X_n) \), and let \( \mathcal{S}^n \) be the set of all possible \( n \)-tuples over \( \mathcal{S} \) that can be attained by \( X^n \). The Shannon--McMillan--Breiman theorem\cite{algoet1988sandwich} states that for any \( \epsilon > 0 \), there exists a set \( \mathcal{E}_\epsilon^{(n)} \subseteq \mathcal{S}^n \) (called the \( \epsilon \)-typical set) such that, with high probability as \( n \to \infty \), most sequences \( x^n \in \mathcal{S}^n \) satisfy
\begin{equation}\label{equationforshannonmcmillanbreimantheorem}
e^{-n(H(P) + \epsilon)} \leq \mathbb{P}(X^n) \leq e^{-n(H(P) - \epsilon)}, 
\end{equation}
Or
\begin{equation}
P\left( \left|-\frac{1}{n}\log \left( X^n\right)-H(P) \right| \right)  \xrightarrow[]{n \to \infty}1,
\end{equation}
and the total probability of the \( \epsilon \)-typical set approaches 1. The complement of \( \mathcal{E}_\epsilon^{(n)} \) is referred to as the non \( \epsilon \)-typical set, consisting of sequences that significantly deviate from the expected empirical entropy. 

\begin{theorem}
    For a stationary ergodic process $\mathcal{X}=(X_1, X_2,...)$ with entropy rate \( H_{rate}(\mathcal{X}) > 0 \), we have
\[
J_{rate}(\mathcal{X}) =  H_{rate}(\mathcal{X}) \textit{ \ almost surely}.
\]
\end{theorem}
\begin{proof}
    For a sufficiently large natural number \( n \), let \( S_n \), \( E_1^n \), and \( E_2^n \) denote the support size of the random vector \( X^n \), the size of the \( \epsilon \)-typical set \( \mathcal{E}_{\epsilon}^{(n)} \), and the size of the non \( \epsilon \)-typical set \( {\mathcal{E}_{\epsilon}^{(n)}}^c \), respectively. From equation (\ref{equationforshannonmcmillanbreimantheorem}), we obtain the bounds 
    \begin{equation*}
        (1-\epsilon)e^{n\left(H(P)-\epsilon \right)}\le E_1^n \le e^{n\left(H(P)+\epsilon \right)},
    \end{equation*}
    and 
    \begin{equation*}
        E_2^n\le \epsilon e^{n\left(H(P)-\epsilon \right)},
    \end{equation*}
    which together implies that
    \begin{equation}
        \begin{split}
            &E_1^n\le S_n \le E_1^n+E_2^n\\
            &(1-\epsilon)e^{n\left(H(P)-\epsilon \right)} \le S_n \le (1-\epsilon)e^{n\left(H(P)+\epsilon \right)}+ \epsilon e^{n\left(H(P)-\varepsilon \right)}\le e^{n\left(H(P)+\epsilon \right)} \\
            & \log(1-\epsilon)+n(H(P)-\epsilon) \le \log(S_n) \le n\left(H(P)+\epsilon \right) \\
            &\frac{\log(1-\epsilon)}{n}+H(P)-\epsilon \le \frac{1}{n} \log(S_n) \le H(P)+\epsilon.
        \end{split}
    \end{equation}
    Since $\epsilon>0$ is fixed so for large $n$, 
    \begin{equation*}
        \frac{\log(1-\epsilon)}{n}\to 0.
    \end{equation*}
    Also, 
    \begin{equation}
        \frac{1}{n}\log(S_n-1)=\frac{1}{n}\log(S_n)+\frac{1}{n}\log\left(1-\frac{1}{S_n}\right),
    \end{equation}
    and for sufficiently large $n$, we have $S_n\to \infty$, which implies that  
    \begin{equation}
    \frac{1}{n}\log(S_n)=\frac{1}{n}\log(S_n-1)-\frac{1}{n}\log\left(1-\frac{1}{S_n}\right) \to \frac{1}{n}\log(S_n-1).
    \end{equation}
    Thus, we have 
    \begin{equation}
       P\left( \left| \frac{1}{n}\log(S_n-1)- H(P) \right|\right) \xrightarrow[]{n \to \infty}1.
    \end{equation}
\end{proof}
This shows that, for a stationary ergodic process \( (X_1, X_2, \ldots) \) consisting of IID discrete random variables, the entropy rate and extropy rate are almost surely equivalent.

\section{Numerical Results}\label{numericalresults}
This section provides illustrative numerical results on extropy rate, organized into different subsections.

\subsection{Information Quantification}

Quantifying information has long been a foundational area of research. Many approaches in this field are grounded in intuitive formulations based on human perception of complexity or simplicity within systems. One such intuition is that the less likely an event is to occur, the more information\cite{shannon1948mathematical} its occurrence conveys. In this context, extropy arises as a complementary measure to entropy, quantifying information in a probability distribution by emphasizing the probabilities of non-occurrence, inspired in part by Shannon’s formulation of entropy.

\begin{figure}[htbp]
    \centering
    \begin{subfigure}[t]{0.468\textwidth}
        \centering
        \includegraphics[width=\linewidth]{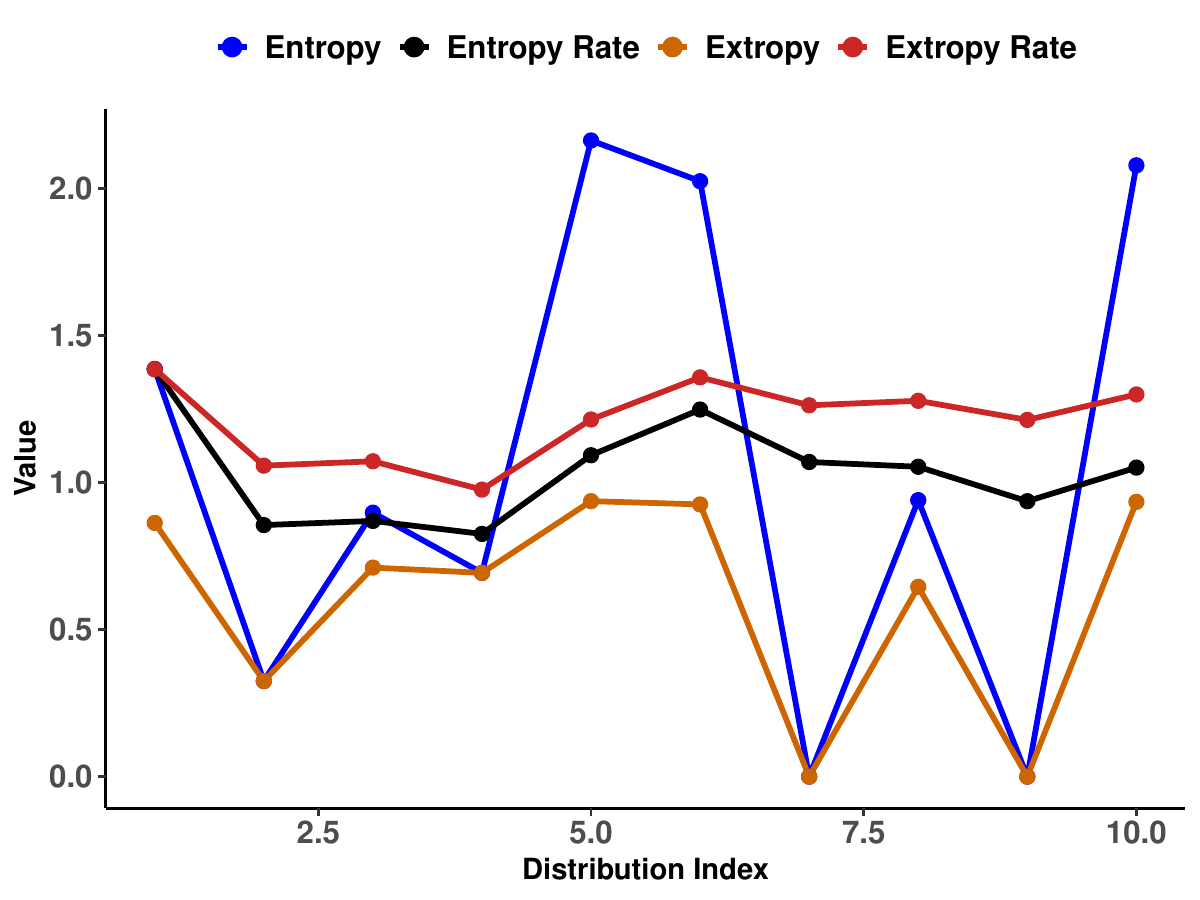}
        \caption{}
        \label{fig:plot1}
    \end{subfigure}
    \hspace{0.05\textwidth}
    \begin{subfigure}[t]{0.468\textwidth}
        \centering
        \includegraphics[width=\linewidth]{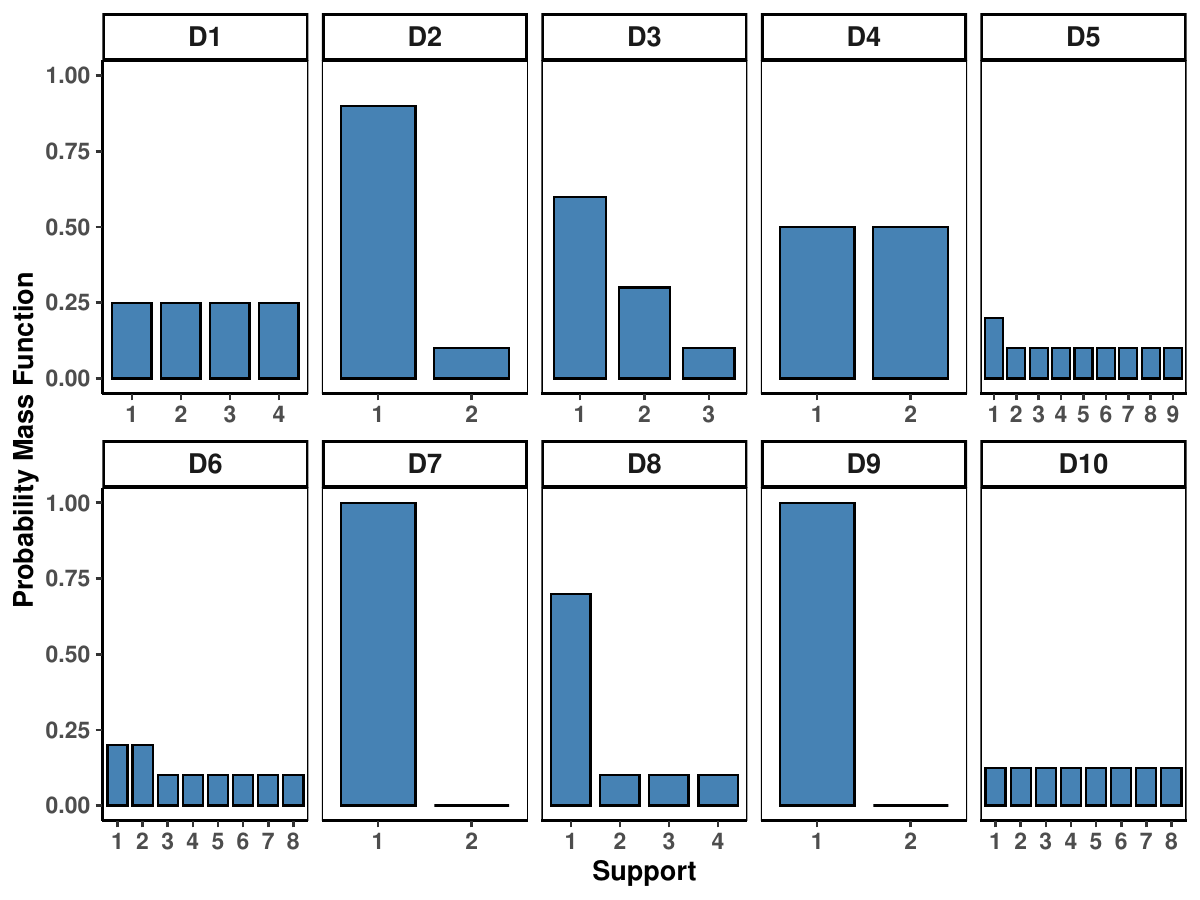}
        \caption{}
        \label{fig:plot2}
    \end{subfigure}
    \caption{(a) Information evolution due to entropy, extropy, entropy rate, and extropy rate in a finite stochastic process consists of $10$ simulated discrete pmfs, and (b) Histograms of those $10$ simulated discrete pmfs.}
    \label{figureInNumerica1}
\end{figure}

\begin{figure}[htbp]
    \centering
    \begin{subfigure}[t]{0.468\textwidth}
        \centering
        \includegraphics[width=\linewidth]{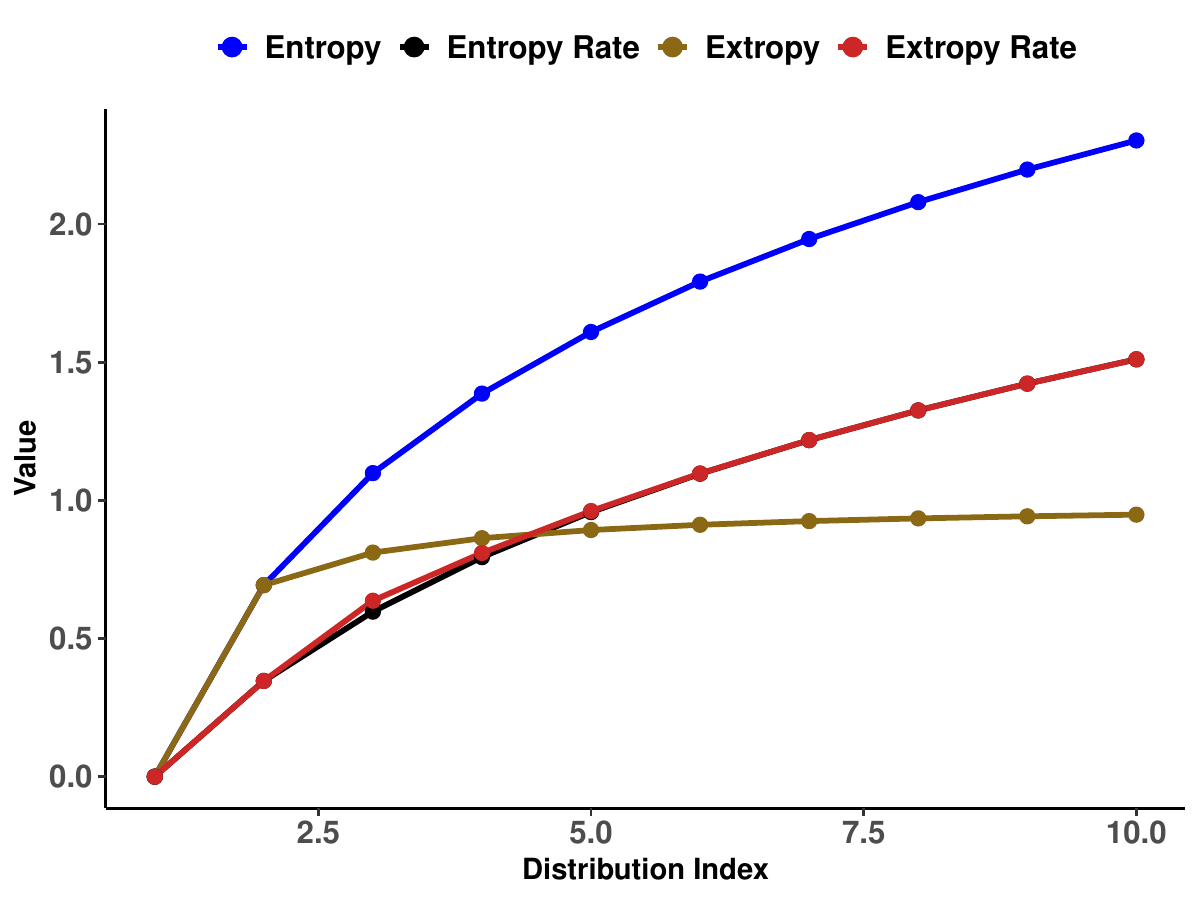}
        \caption{}
        \label{fig:plot1}
    \end{subfigure}
    \hspace{0.05\textwidth}
    \begin{subfigure}[t]{0.468\textwidth}
        \centering
        \includegraphics[width=\linewidth]{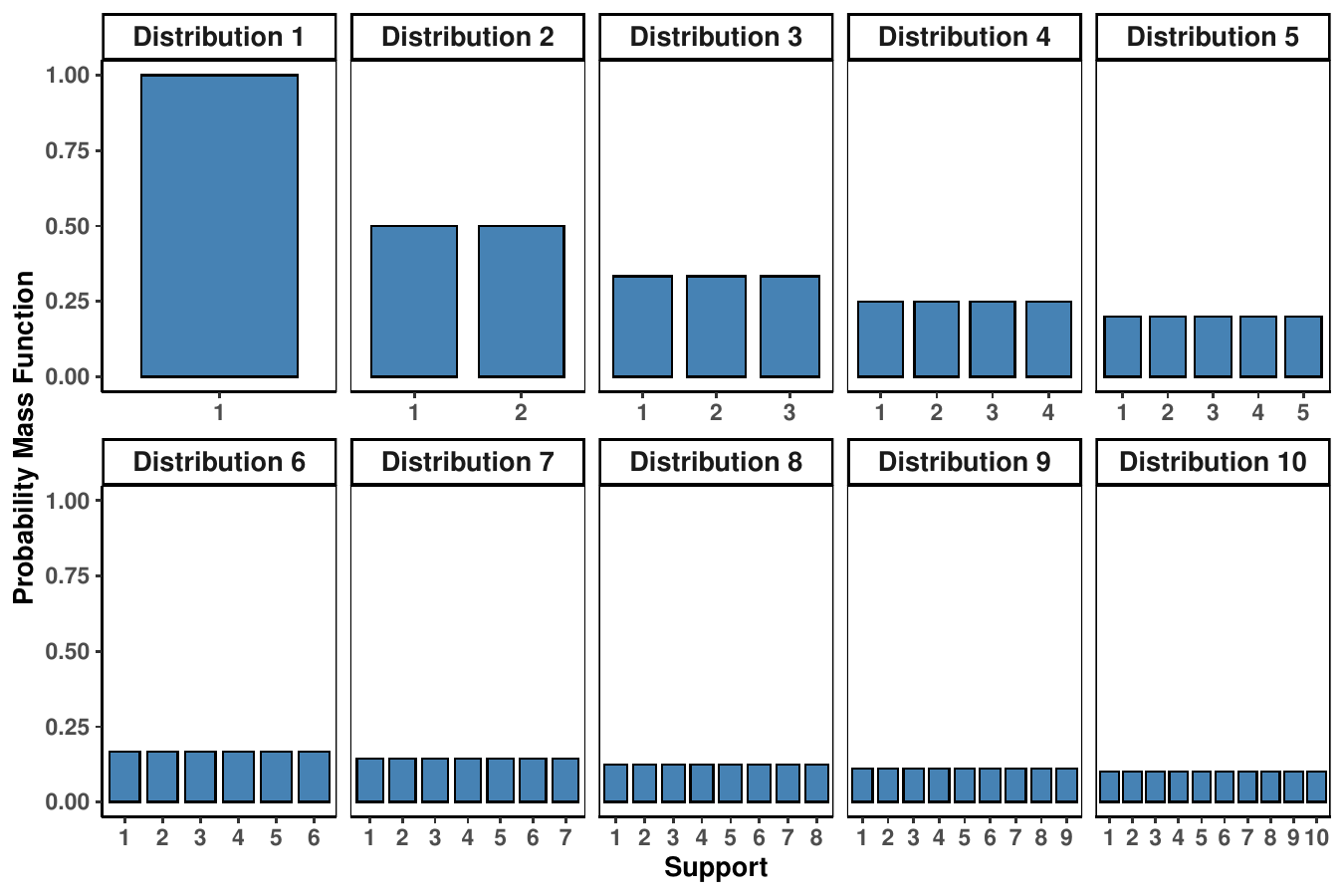}
        \caption{}
        \label{fig:plot2}
    \end{subfigure}
    \caption{(a) Information evolution due to entropy, extropy, entropy rate, and extropy rate in a finite stochastic process consists of $10$ different uniform pmfs, and (b) Histograms of those $10$ pmfs.}
    \label{figureInNumerica2}
\end{figure}

We propose the extropy rate as the measure average amount of information associated with each random variable in a discrete-time stochastic process. While the entropy rate\cite{cover1999elements} is a well-established measure of average uncertainty per random variable, the extropy rate provides an alternative perspective by focusing on complementary probabilities. To illustrate this, we simulate $10$ distinct probability mass functions (denoted as $D1, D2,...,D10$), whose histograms are shown in Figure \ref{figureInNumerica1}. These pmfs are deliberately chosen to be diverse in shape and complexity. For each distribution, we compute the entropy and extropy values and the corresponding entropy rate $H_{rate}(X_1, ..., X_i)$, and the extropy rate $J_{rate}(X_1, ..., X_i)$ for the joint random variables $(X_1, ..., X_i)$ are calculated for $i = 1, 2, ..., n$. Figure \ref{figureInNumerica1} demonstrates how the extropy rate fluctuates with varying entropy and extropy values across individual distributions, effectively capturing the informational contribution of each random variable. We further illustrate this behaviour using a sequence of $10$ uniform distributions with support sizes increasing from $1$ to $10$. As shown in Figure \ref{figureInNumerica2}, as the support size and complexity increase, the extropy rate also increases, reflecting a greater average information content per variable. Also, figure \ref{figureInNumerica2} illustrates that the entropy rate and extropy rate for a uniformly distributed stochastic process become indistinguishable past some time steps, demonstrating asymptotic equivalence. Notably, as observed by Lad et al.\cite{lad2015extropy}, when the support size is $2$, both the entropy and extropy functionals coincide, yielding identical values. Consequently, the entropy and extropy rates behave similarly in such cases. This demonstrates that the extropy rate is capable of quantifying the average information per random variable in a stochastic process and reveals their dependency structure through the information shared among them.

\subsection{Time-Series Complexity Measure}
Motivated by the definition of entropy rate, numerous measures have been proposed to distinguish time series based on their inherent complexity. Among these, approximate entropy and permutation entropy are particularly prominent. Approximate entropy\cite{delgado2019approximate}
depends on the selection of block size and tolerance level and is computed using the probabilities of similar blocks within the series. In contrast, permutation entropy\cite{bandt2002permutation} relies on two parameters, embedding time delay and embedding dimension, which estimates the probabilities of ordinal patterns formed by segments of the time series, thereby capturing its fluctuations.

To estimate the extropy rate for a given finite time-series sample $(X_1, \dots, X_n)$, we begin by estimating the joint probability mass function $P(X_1, \dots, X_n)$. Assuming the observations are independent in time series, we approximate this joint probability as the product of marginal probabilities $P(X_1) \cdots P(X_n)$. Each marginal probability $P(X_j)$ is estimated empirically by the relative frequency of $X_j$ in the sample, i.e., $P(X_j) = a_j / n$, where $a_j$ is the count of $X_j$ in the series, and $n$ is the total number of observations. The total support size is estimated by $m^n$, where $m$ is the number of distinct values observed in the series. Using these quantities, we estimate the extropy rate as
\begin{equation}\label{estimatorofextropyrate}
\hat{J}_{rate}(X_1,X_2,...,X_n)=\frac{1}{n}\left(\log(m^n-1)-\frac{(1-P(X_1,X_2,....,X_n))\log(1-P(X_1,X_2,....,X_n))}{m^n-1} \right).
\end{equation}

\begin{figure}[H] % Use [H] to place the figure exactly here, or [htbp] for flexible placement
    \centering
    \includegraphics[width=0.85\textwidth]{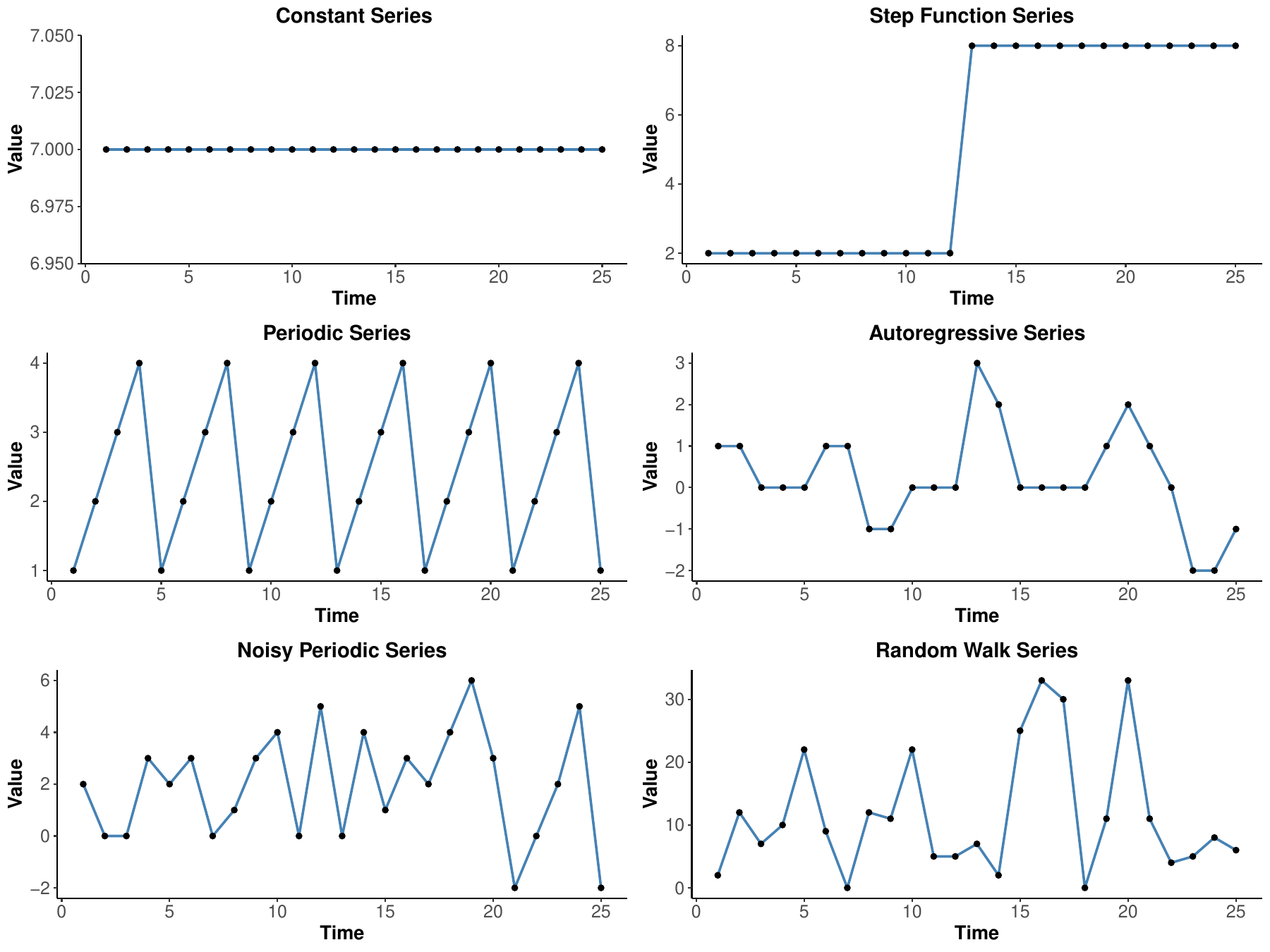} % Path to your image
    \caption{Time-series plots for samples of Constant, Step Function, Periodic, Autoregressive, Noisy Periodic, and Random Walk processes over $25$ time steps.}
    \label{figuretime-series}
\end{figure}

\begin{table}[ht]
\centering
\begin{tabular}{|c|l|c|c|c|}
\hline
S. No. & Time-Series        & Approximate entropy & Permutation entropy & Extropy Rate \\
\hline
$1$      & Constant           & $0.00$ &     $0.00$     & $\textbf{0.00}$            \\
$2$      & Step Function      & $0.15$  &    $0.14$      &  $\textbf{1.00}$           \\
$3$      & Periodic           & $0.25$   &     $0.21$    & $\textbf{2.00}$            \\
$4$      & Autoregressive     & $0.32$    &    $0.43$    & $\textbf{2.58}$           \\
$5$      & Noisy Periodic     & $0.54$     &    $0.44$   & $\textbf{3.00}$            \\
$6$      & Random Walk        & $0.61$      &   $0.45$     & $\textbf{3.90}$            \\
\hline
\end{tabular}
\caption{Approximate entropy, Permutation entropy and extropy rate for six different types of time series.}
\label{table-timeseries}
\end{table}

In this demonstration, we generated six finite time series, each of length $25$, representing diverse patterns, which are constant, step function, periodic, autoregressive, noisy periodic, and random walk. These variations are illustrated in Figure \ref{figuretime-series}, clearly reflecting the diversity in their dynamics. Approximate entropy and permutation entropy are widely used non-parametric measures for quantifying the complexity of finite time series. Accordingly, in addition to estimating the extropy rate, we also estimated the approximate and permutation entropy values using the RStudio functions \texttt{approx\_entropy} and \texttt{permutation\_entropy} from the packages \texttt{pracma} and \texttt{StatComp}, respectively.
The results in Table \ref{table-timeseries} indicate that the constant time series exhibits the lowest complexity, while the random walk displays the highest complexity among the six samples. It is evident that the extropy rate follows a similar trend to both approximate and permutation entropy, indicating its effectiveness in capturing the complexity of time series data.

% \begin{table}[ht]
% \centering
% \begin{tabular}{|c|l|c|c|}
% \hline
% S. No. & Time-Series        & Approximate entropy & Extropy Rate \\
% \hline
% $1$      & Constant           & $0.00$           & $\textbf{0.00}$            \\
% $2$      & Step Function      & $3.03\times{10}^{-8}$            &  $\textbf{1.00}$           \\
% $3$      & Periodic           & $1.88\times{10}^{-15}$            & $\textbf{2.00}$            \\
% $4$      & Autoregressive     & $1.39\times{10}^{-16}$            & $\textbf{2.58}$           \\
% $5$      & Noisy Periodic     & $2.04\times{10}^{-21}$            & $\textbf{3.00}$            \\
% $6$      & Random Walk        & $1.61\times{10}^{-23}$            & $\textbf{3.32}$            \\
% \hline
% \end{tabular}
% \caption{Entropy and extropy rate for six different types of time series.}
% \label{table-timeseries}
% \end{table}

\subsection{Characterizing the Chaotic behavior in Nonlinear Maps}

Consider the logistic and Hénon maps\cite{elaydi2007discrete}, given by 

\begin{equation}
x_{n+1}=rx_n(1-x_n), \textit{ \ \ where \ $r\in [2.5,4]$ \ and $0<x_n<1$ \ for all $n=0,1,2,...,$}
\end{equation}
and 
\begin{equation}
\begin{split}
    x_{n+1}&=1+y_n-ax_n^2 \\
    y_{n+1}&=bx_n, \textit{ \ where $1<a<1.4$ and $b=0.3$ are parameters,}
\end{split}
\end{equation}

respectively. These are examples of deterministic nonlinear dynamical systems, where the state of the system evolves according to fixed rules. For a given initial points and bifurcation parameter, the system may settle into stable fixed points or periodic cycles, allowing for short-term predictability. As the bifurcation parameter varies, these maps can undergo bifurcations, which are qualitative shifts in behaviour where a single stable state splits into multiple, eventually transitioning into chaos. Analyzing such transitions is essential for understanding the system's predictability and stability. Extropy rate can indicate the bifurcation by exhibiting sudden spikes precisely at the moments of transition, providing an interpretable signal of emerging complexity in the system's dynamics.

We plotted the bifurcation diagram for the logistic map in figure \ref{fig:logistic_summary}(a), where the parameter $r$ varies from $2.5$ to $4$ with the initial condition $x_0 = 0.1$. Similarly, for the Hénon map with $b = 0.3$, a ranging from $1$ to $1.4$, and initial point $(x_0, y_0) = (0.1, 0.1)$, the corresponding bifurcation diagram is shown in figure \ref{fig:logistic_summary}(a). Using the extropy rate estimator defined in equation (\ref{estimatorofextropyrate}), we computed the extropy rate for each parameter value ($r$ for the logistic map and $a$ for the Hénon map). We plotted the corresponding extropy rate values in figures \ref{fig:logistic_summary}(b) and \ref{fig:Henon_summary}(b), respectively. It is evident that at each bifurcation point, the extropy rate exhibits a sharp increase, effectively signalling the onset of complex dynamics. For example, in figure \ref{fig:logistic_summary}, notable jumps in extropy are observed around $r = 3, 3.5$, and near $3.8$, while in figure \ref{fig:Henon_summary}, similar behaviour is seen at $a = 1.02, 1.2$, and $1.3$. These observations suggest that the extropy rate can serve as a reliable indicator of bifurcation.

\begin{figure}[H]
    \centering

    % First plot: Bifurcation
    \begin{subfigure}[t]{0.45\textwidth}
        \centering
        \includegraphics[width=\textwidth]{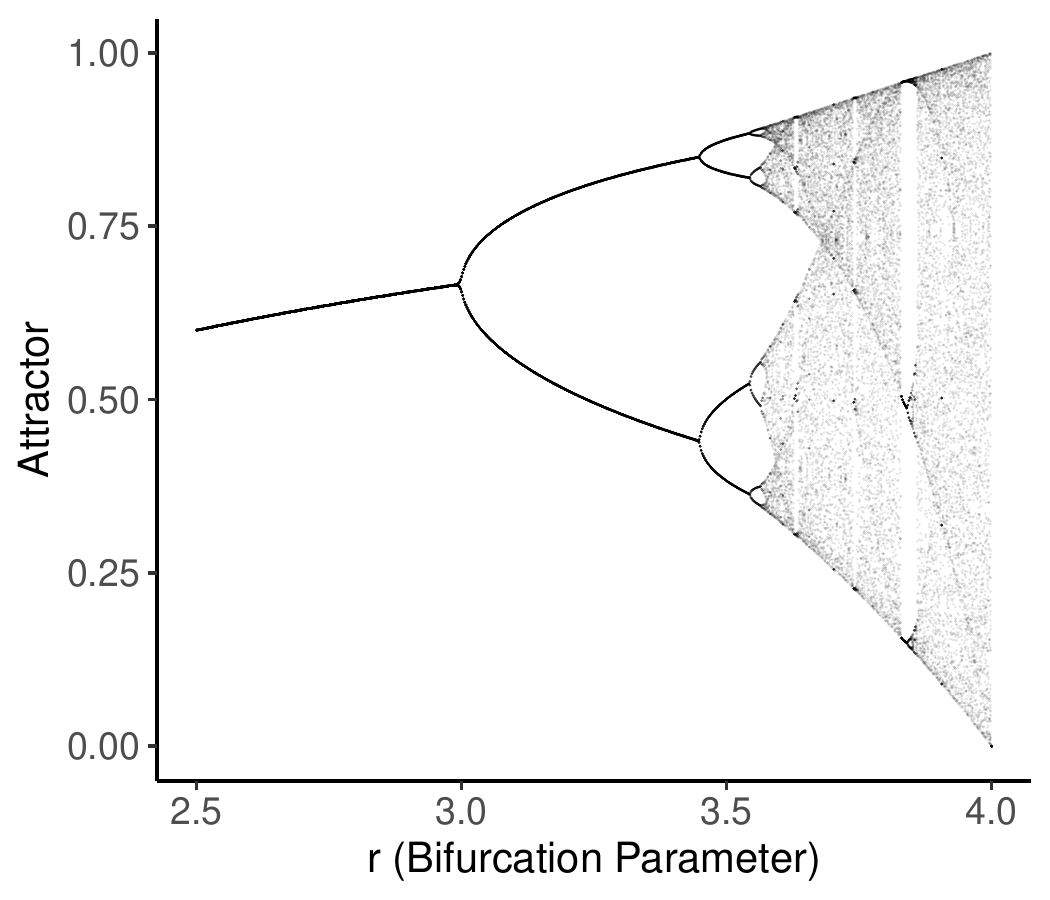}
        \caption{Bifurcation Diagram}
    \end{subfigure}
    \hfill
    % Second plot: Permutation Entropy
    % \begin{subfigure}[t]{0.31\textwidth}
    %     \centering
    %     \includegraphics[width=\textwidth]{logisticPermutationEntropy.pdf}
    %     \caption{Permutation Entropy}
    % \end{subfigure}
    % \hfill
    % Third plot: Extropy Rate
    \begin{subfigure}[t]{0.45\textwidth}
        \centering
        \includegraphics[width=\textwidth]{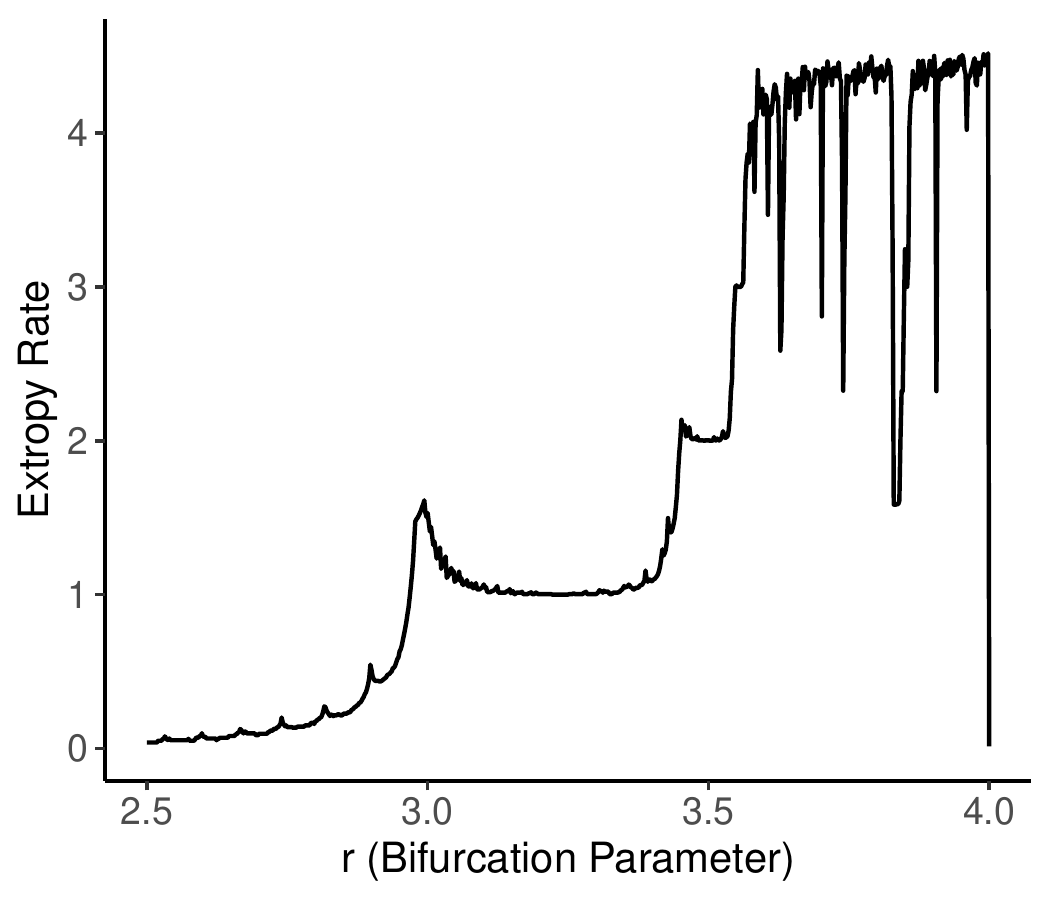}
        \caption{Extropy Rate}
    \end{subfigure}

    \caption{Logistic map bifurcation diagram and extropy rate across logistic map parameter \(r\).}
    \label{fig:logistic_summary}
\end{figure}

\begin{figure}[H]
    \centering

    % First plot: Bifurcation
    \begin{subfigure}[t]{0.45\textwidth}
        \centering
        \includegraphics[width=\textwidth]{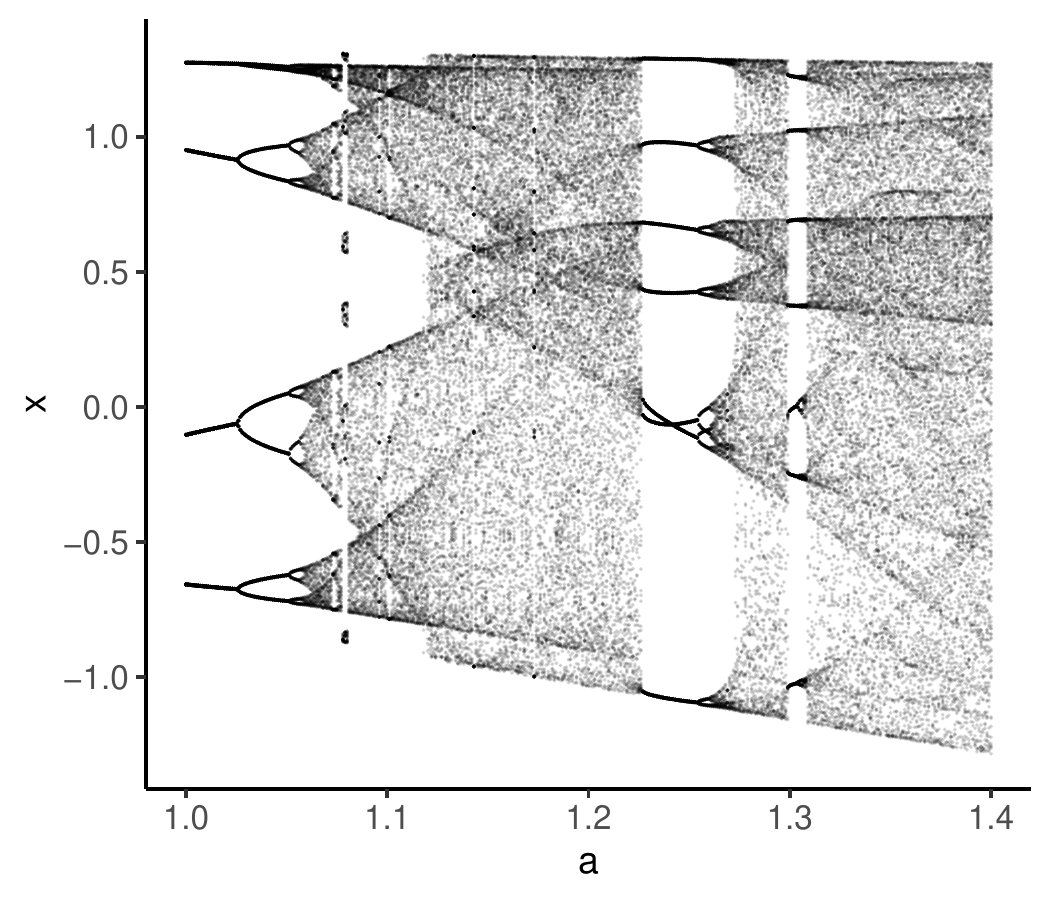}
        \caption{Bifurcation Diagram}
    \end{subfigure}
    \hfill
    % Second plot: Permutation Entropy
    % \begin{subfigure}[t]{0.31\textwidth}
    %     \centering
    %     \includegraphics[width=\textwidth]{HenonPermutationentropy.pdf}
    %     \caption{Permutation Entropy}
    % \end{subfigure}
    % \hfill
    % Third plot: Extropy Rate
    \begin{subfigure}[t]{0.45\textwidth}
        \centering
        \includegraphics[width=\textwidth]{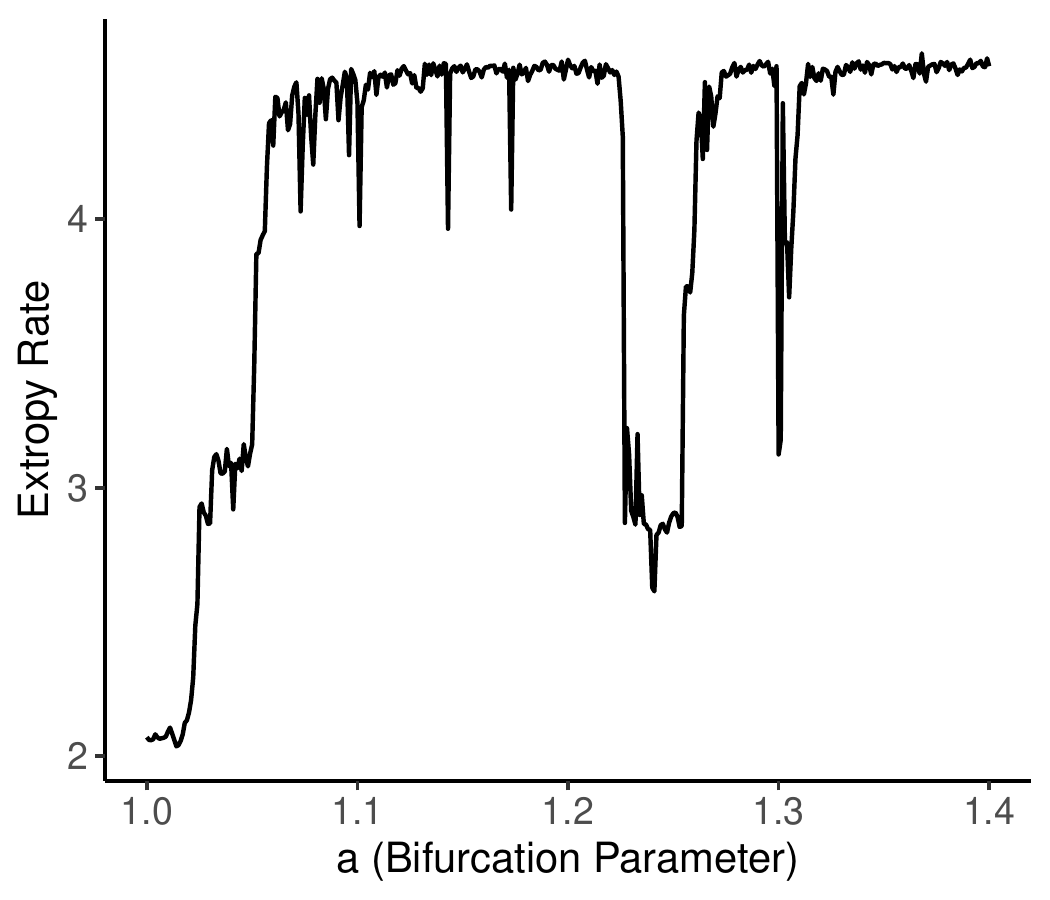}
        \caption{Extropy Rate}
    \end{subfigure}

    \caption{Henon map bifurcation diagram and extropy rate across Henon map parameter \(r\).}
    \label{fig:Henon_summary}
\end{figure}

\subsection{Relation with Simpson's Diversity Index}
Simpson’s Diversity Index (SDI)\cite{morris2014choosing} quantifies the diversity associated with a discrete probability distribution $P=\{p_i\}$, and is given by 
\begin{equation}
\text{SDI}(P)=1-\sum_{i}p^2_i.
\end{equation}

\begin{figure}[htbp]
    \centering
    \includegraphics[width=0.85\textwidth]{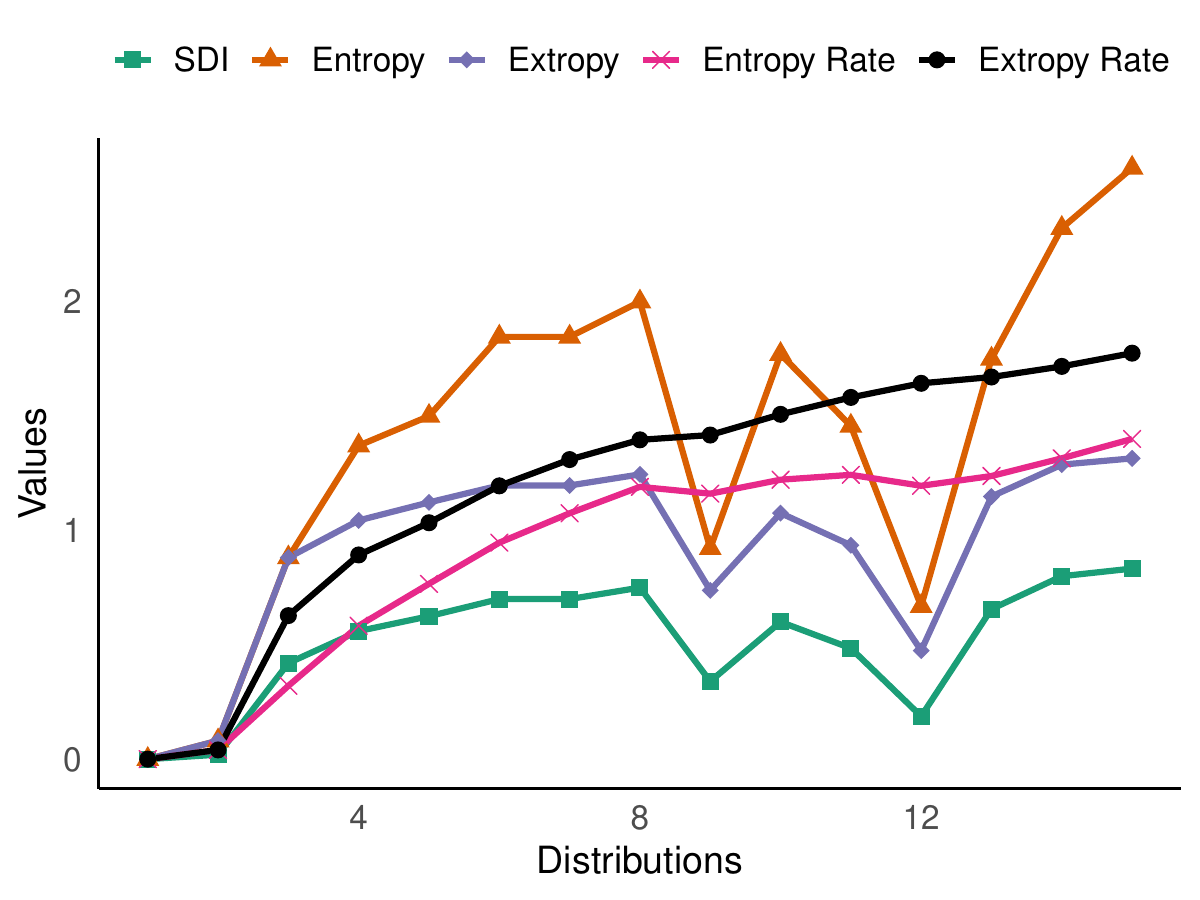}
    \caption{Simpson's Diversity Index, Entropy, Entropy Rate, Extropy, and Extropy Rate for $15$ generated discrete distributions.}
    \label{fig:diversity-metrics}
\end{figure}

Its value lies between $0$ and $1$, where $0$ indicates minimal diversity, as in the case of a degenerate distribution and values closer to $1$ reflect higher diversity. To highlight the relevance of the extropy rate, we generated $15$ diverse independently distributed discrete probability distributions corresponding to random variables $X_1, \dots, X_{15}$. For each distribution, we computed the SDI, Shannon entropy, and extropy individually. Additionally, for the finite stochastic process $(X_1, \dots, X_i)$ involving the first $i=1,2,...,15$ random variables, we computed the entropy rate and extropy rate, given by 
\begin{equation}
H_{rate}((X_1, \dots, X_i))=-\frac{\sum_{j_1, j_2,...,j_i}p_{j_1, j_2,...,j_i}\log(p_{j_1, j_2,...,j_i})}{i} ,
\end{equation}
and 
\begin{equation}
J_{rate}((X_1, \dots, X_i))=-\frac{\sum_{j_1, j_2,...,j_i}\left(1-p_{j_1, j_2,...,j_i}\right)\log\left(1-p_{j_1, j_2,...,j_i}\right)}{i} 
\end{equation}
respectively. Figure \ref{fig:diversity-metrics} shows the calculated values of SDI, entropy, extropy, entropy rate, and extropy rate. It is evident from the figure \ref{fig:diversity-metrics} that SDI aligns more closely with extropy than with entropy. This pattern persists across a wide range of discrete distributions we tested, although no precise mathematical relationship, such as a scaling or translation transformation, has yet been established to connect SDI and extropy directly. This empirical similarity suggests that SDI is more sensitive to the probabilities of non-occurrence than to occurrence, a property more naturally captured by extropy.
Therefore, in applications such as ecological systems, where monitoring the evolution of diversity over time is important, the extropy rate offers a more suitable framework than the entropy rate. By averaging over time, the extropy rate effectively captures how diversity changes across successive states of the system, making it a compelling tool for tracking temporal diversity patterns.

\section{Applications of extropy rate}\label{applications}
In this section, we present an application of the extropy rate in feature selection and demonstrate its applicability on six publicly accessible datasets.

\subsection{Feature Selection}
Feature selection is a crucial step in machine learning that involves identifying the most relevant features in a dataset to improve model performance, reduce overfitting, and improve interpretability. By removing irrelevant or redundant features, it simplifies the model, lowers computational costs, and helps prevent the curse of dimensionality, especially in high-dimensional datasets. Common feature selection methods are mutual information, chi-square, and F-score\cite{yang1997comparative}. The mutual information method selects features that share the most information with the target variable. The chi-square method uses the chi-square statistic to evaluate the dependence between each feature and the target, selecting those with high values. The F-score method evaluates how well each feature can differentiate between classes in the target variable.
\subsection{Methodology}

It is well known that a probability distribution with higher uncertainty carries more information than one with lower uncertainty. We apply this idea to guide feature selection using the extropy rate. We estimate the extropy rate for a given set of features based on their joint probability distribution. The extropy rate is empirically estimated by assigning uniform probabilities to all observed states of the random variable and adjusting for repeated values to form an empirical probability distribution. These extropy rate values capture the level of dependency among features and reveal how uncertainty evolves over time. To select $k$ features out of $n$, we identify the $k$ features for which the estimated extropy rate is highest. For instance, if the extropy rate estimator based on the feature subset $(X_1, \ldots, X_j)$ is highest, then feature $X_j (j = 1, 2, \ldots, n)$ is selected, as illustrated in flowchart in figure \ref{fig:featureselectionmethodology}. This approach selects features that carry the most information while accounting for inter-feature dependencies.

\begin{figure}[htbp]
    \centering
    \includegraphics[width=0.90\textwidth]{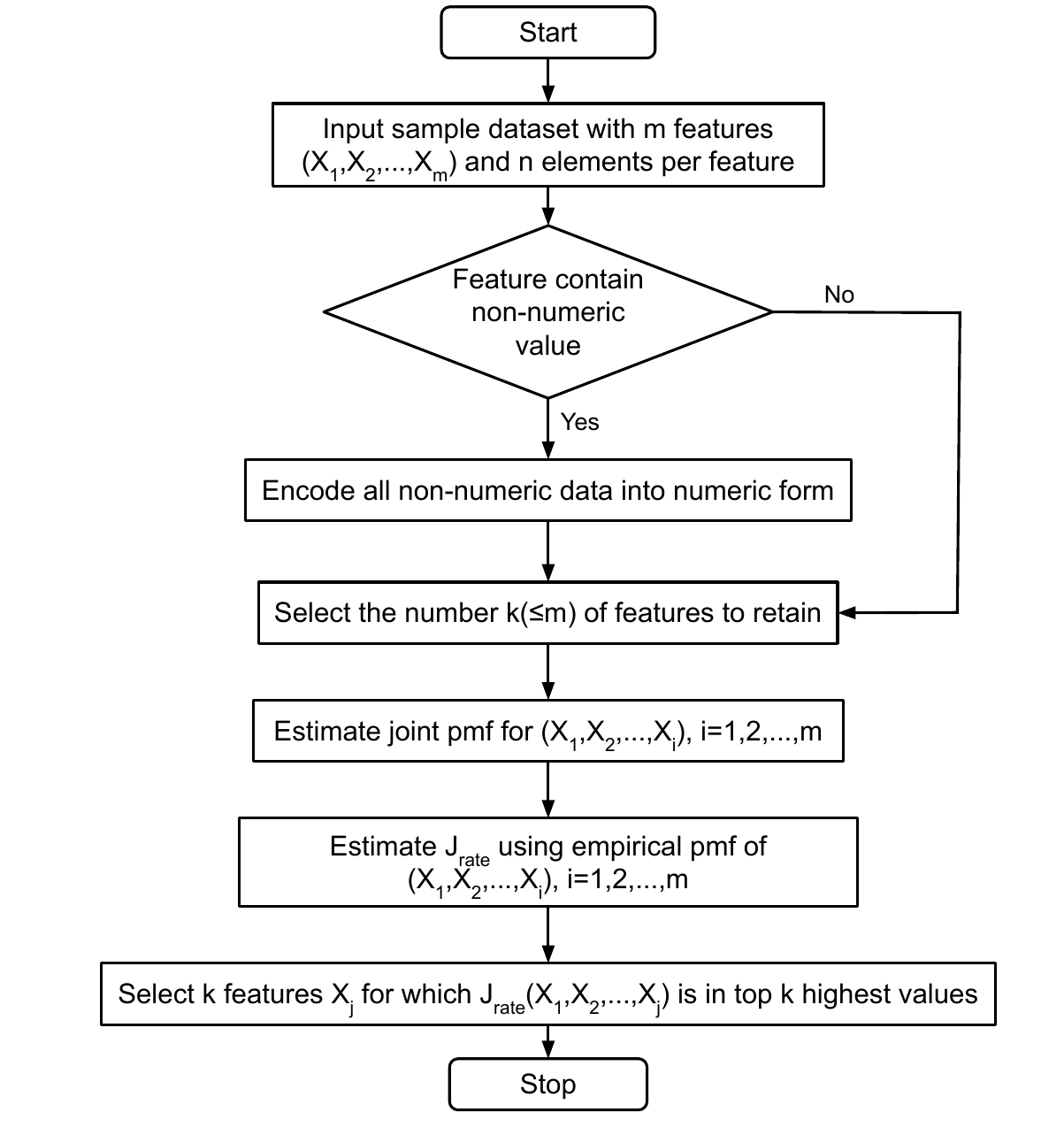}
    \caption{Feature selection methodology using the extropy rate estimates}
    \label{fig:featureselectionmethodology}
\end{figure}

%%%%%%%%%%%%%%%%%%%%%%%%%%%%%%%%%%
\subsection{Dataset}
We apply the proposed feature selection methodology to six well-known datasets, using only their numerical features. The datasets are as follows: the first is the diabetes dataset\cite{diabetesdataset}, which contains ten numerical features; the second is the blood transfusion dataset\cite{bloodtransfusiondataset} with four numerical variables; the third is Boston housing dataset\cite{BostonHousingdataset} includes eleven numerical features; the fourth is a sample of eye EEG measurements\cite{eegeyestatedataset} with five numerical features; the fifth relates to forest fire data\cite{forestfiredataset} and contains ten numerical features; and the sixth is an electricity demand dataset\cite{electricitydataset} with seven numerical features and one target variable. The source links to the datasets are also provided in Table \ref{tab:feature_selection_comparison} to their source sites.

\begin{figure}[htbp]
    \centering

    % Row 1
    \begin{subfigure}[b]{0.45\textwidth}
        \centering
        \includegraphics[width=\textwidth]{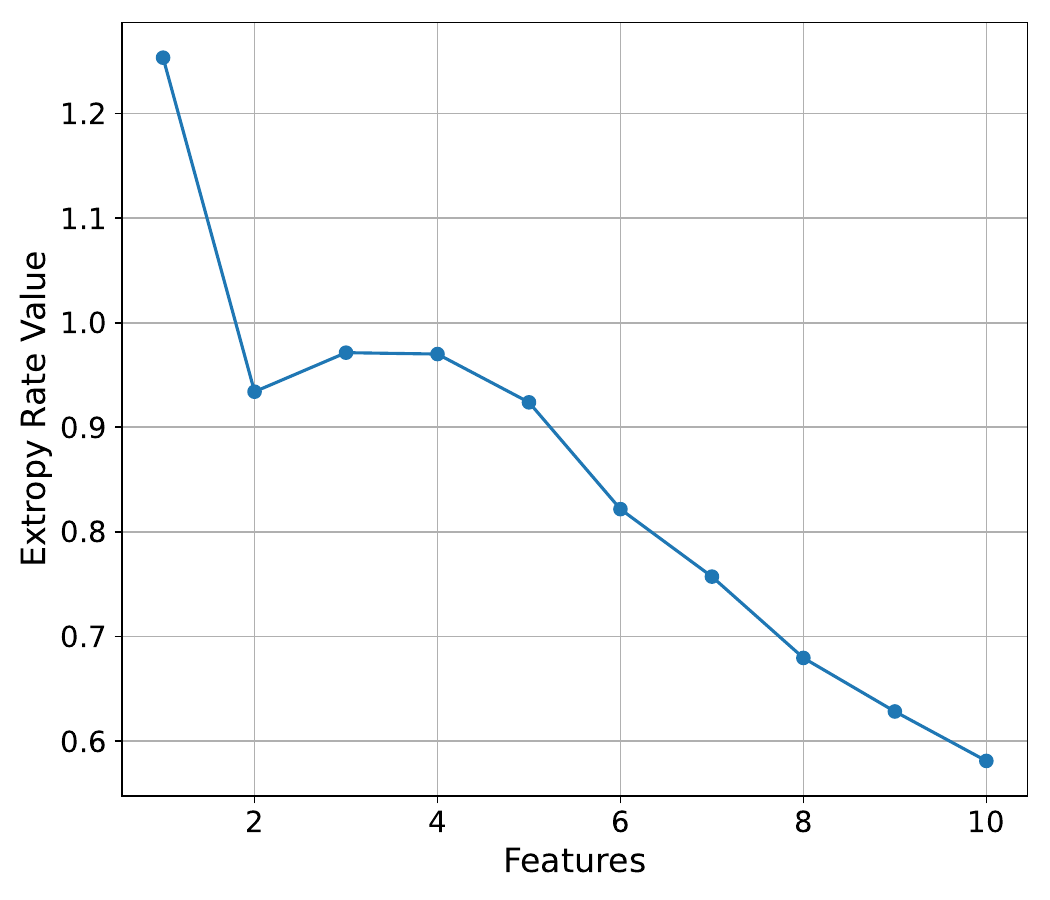}
        \caption{Diabetes}
    \end{subfigure}
    \hfill
    \begin{subfigure}[b]{0.45\textwidth}
        \centering
        \includegraphics[width=\textwidth]{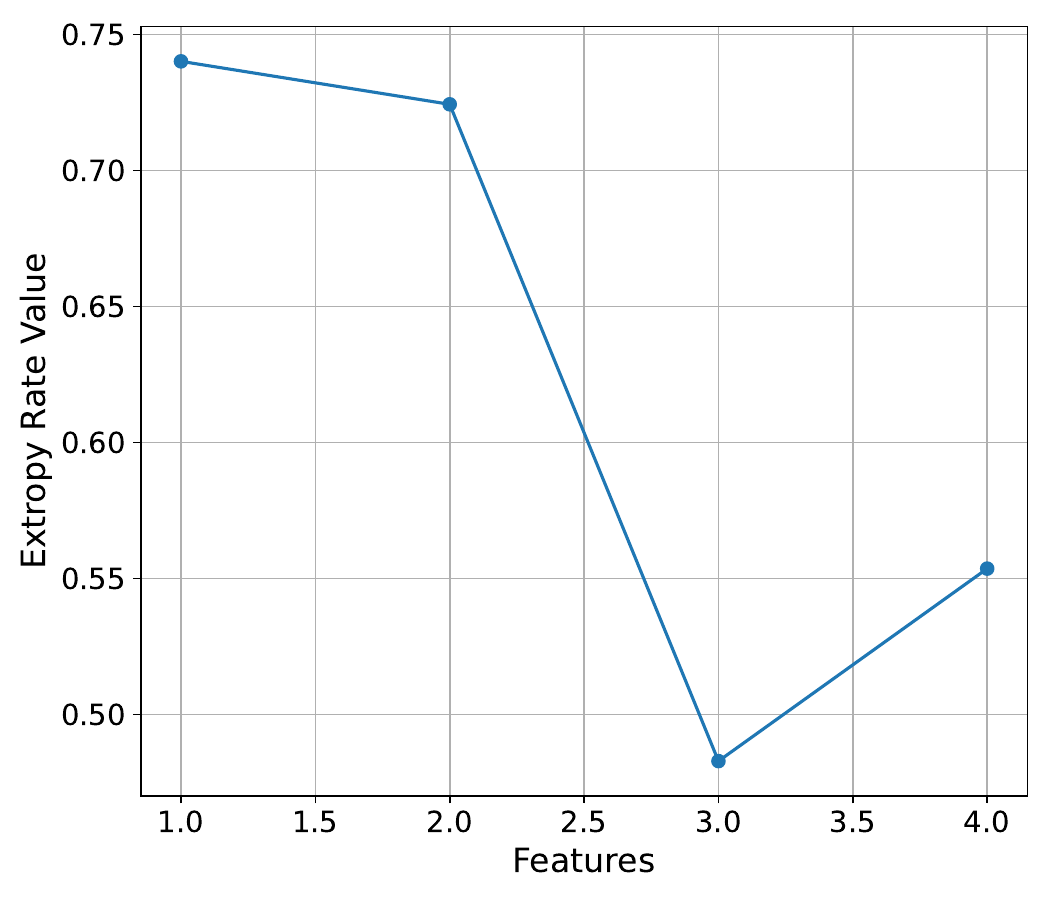}
        \caption{Blood Transfusion}
    \end{subfigure}

    \vspace{1em}

    % Row 2
    \begin{subfigure}[b]{0.45\textwidth}
        \centering
        \includegraphics[width=\textwidth]{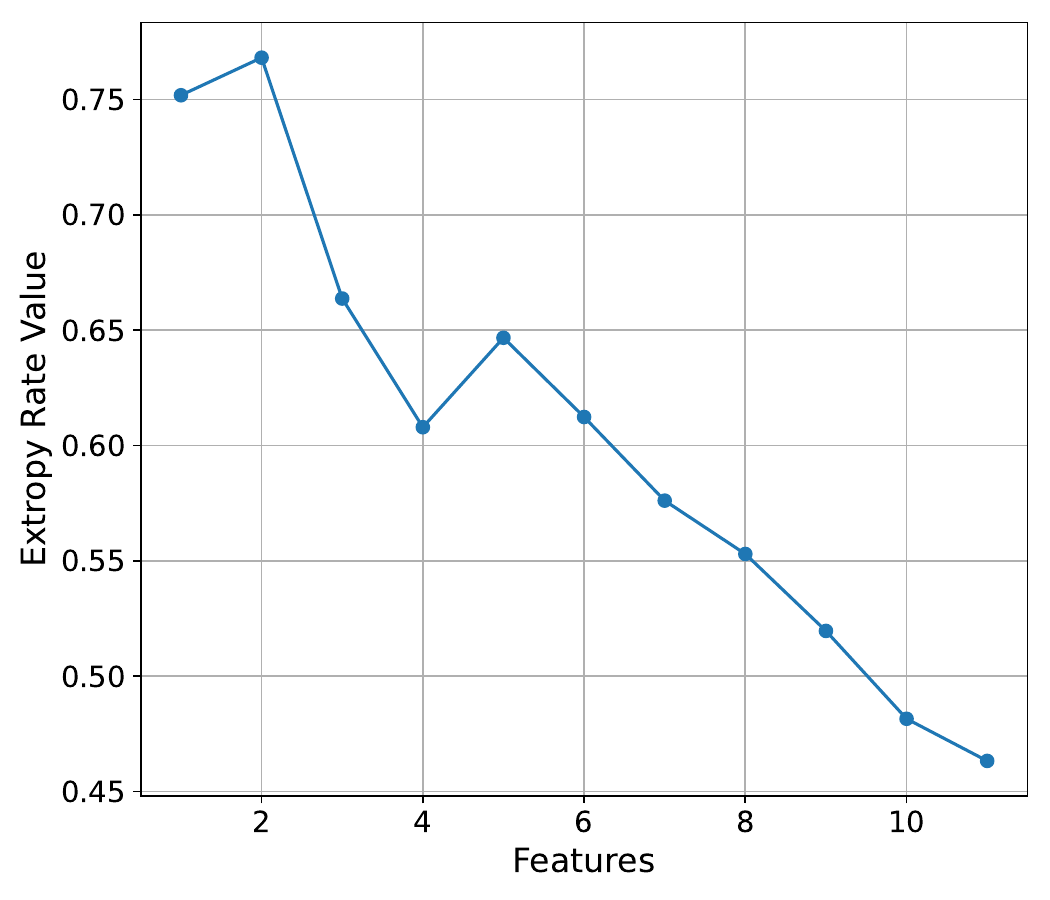}
        \caption{Boston Housing}
    \end{subfigure}
    \hfill
    \begin{subfigure}[b]{0.45\textwidth}
        \centering
        \includegraphics[width=\textwidth]{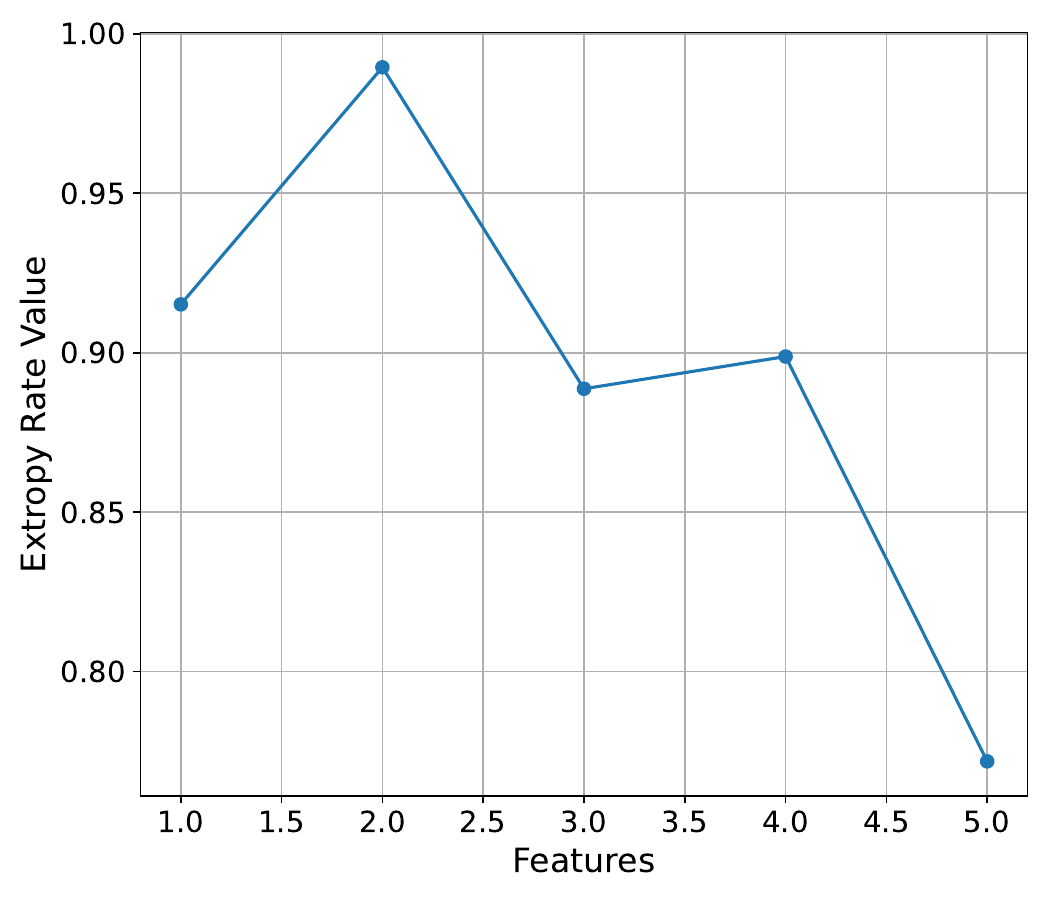}
        \caption{EEG eye state}
    \end{subfigure}

    \vspace{1em}

    % Row 3
    \begin{subfigure}[b]{0.45\textwidth}
        \centering
        \includegraphics[width=\textwidth]{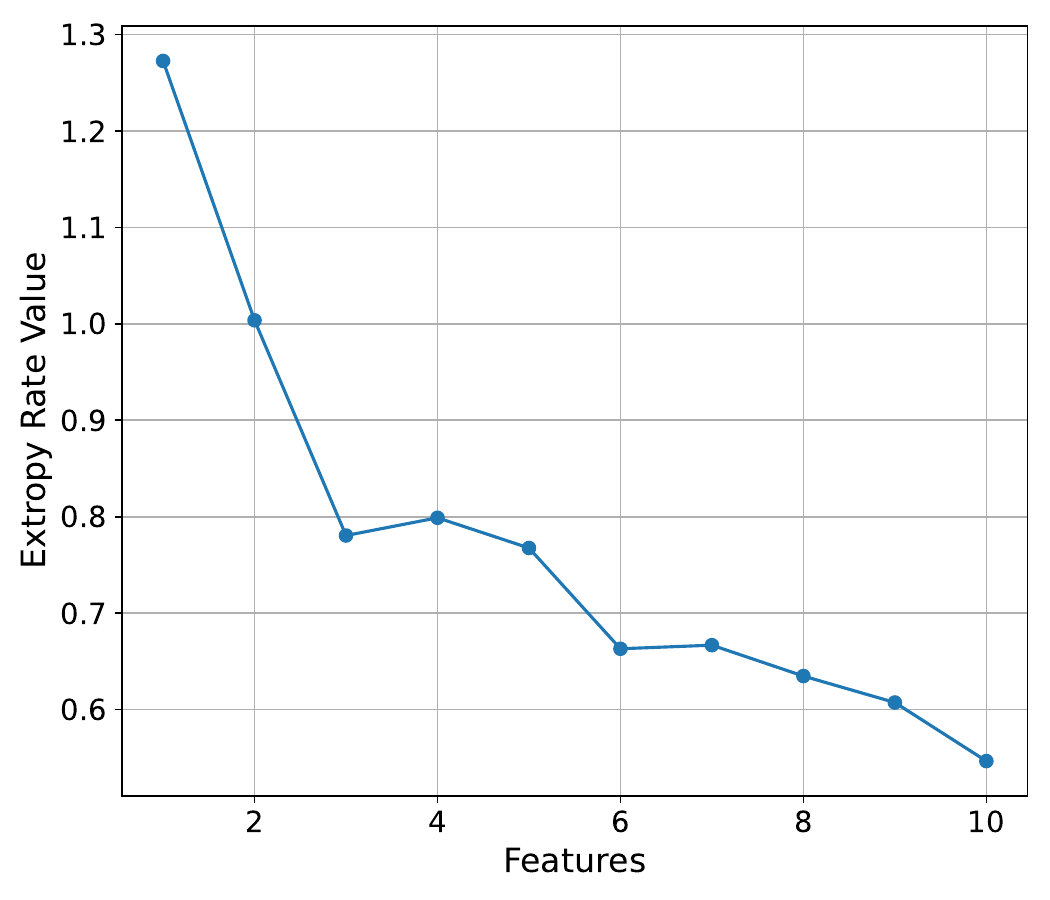}
        \caption{Forest Fire}
    \end{subfigure}
    \hfill
    \begin{subfigure}[b]{0.45\textwidth}
        \centering
        \includegraphics[width=\textwidth]{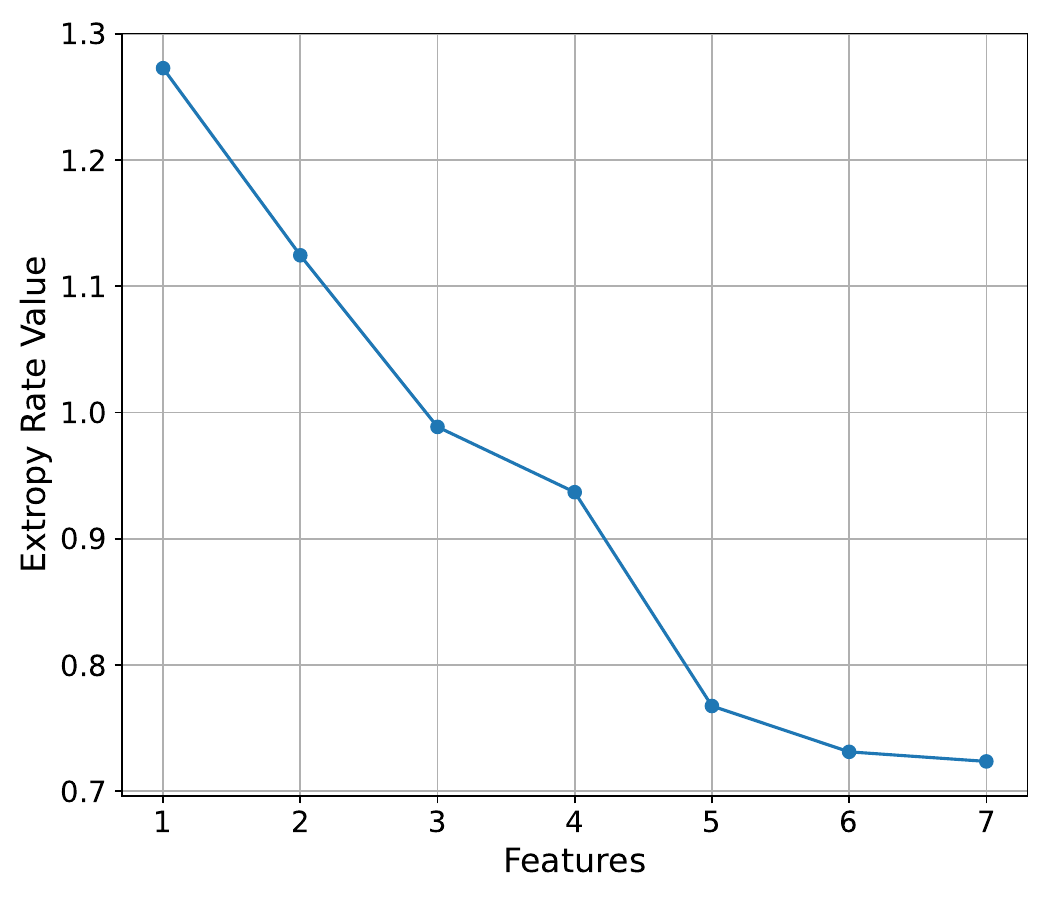}
        \caption{Electricity}
    \end{subfigure}

    \vspace{1em}

    \caption{Estimated extropy rate values based on sequential feature inclusion for the Diabetes, Blood Transfusion Service Center, Boston Housing, EEG Eye State, Forest Fire, and Electricity datasets.}
    \label{figureofextropyrates}
\end{figure}

%%%%%%%%%%%%%%%%%%%%%%%%%%%%%%%%%%
\begin{table}[htbp]

\renewcommand{\arraystretch}{1.5} % Increases row height for readability
\centering
\begin{tabular}{>{\raggedright\arraybackslash}p{2.5cm}>{\centering\arraybackslash}p{1.4cm}>{\raggedright\arraybackslash}p{2.3cm}>{\centering\arraybackslash}p{1.4cm}>{\centering\arraybackslash}p{1.4cm}>{\centering\arraybackslash}p{1.4cm}}
\hline

\textbf{Dataset Name} & \textbf{No. of Features} & \textbf{Method} & \textbf{Accuracy} & \textbf{F1-score} & \textbf{TPR} \\
\hline
 \multirow{4}{*}{\href{https://www4.stat.ncsu.edu/~boos/var.select/diabetes.html}{Diabetes}
} & \multirow{4}{*}{$3$} & Extropy Rate &  $\textbf{1.0000}$     &  $\textbf{1.0000}$     &  $\textbf{1.0000}$    \\ 
                         % Row 2
     &       &   Mutual Info    &   $0.8258$    &   $0.8126$    & $\textbf{1.0000}$      \\ % Row 3
   &        &   Chi-square     &   $0.9661$    &  $0.9633$     &    $\textbf{1.0000}$   \\ % Row 4
     &  & F-score &   $\textbf{1.0000}$   &   $\textbf{1.0000}$    &    $\textbf{1.0000}$           \\ % Row 5
\hline

 \multirow{4}{*}{\parbox{2.5cm}{\href{https://www.openml.org/search?type=data&sort=runs&id=1464&status=active}{Blood Transfusion Service Center}}} & \multirow{4}{*}{$3$} & Extropy Rate &  $\textbf{0.9318}$     &  $\textbf{0.9289}$     &  $\textbf{0.7528}$    \\ 
                         % Row 2
     &       &   Mutual Info    &  $0.8369$     &   $0.8131$    & $0.3876$       \\ % Row 3
   &        &   Chi-square     &    $0.8369$   &   $0.8131$    &     $0.3876$  \\ % Row 4
     &  & F-score &   $0.8369$   &   $0.8131$   &             $0.3876$ \\ % Row 5
\hline

\multirow{4}{*}{\href{https://www.kaggle.com/datasets/abhijithudayakumar/the-boston-housing-dataset}{Boston Housing}
} & \multirow{4}{*}{$4$} & Extropy Rate &  $\textbf{1.0000}$     &  $\textbf{1.0000}$     &  $\textbf{1.0000}$    \\ 
                         % Row 2
     &       &   Mutual Info    &   $0.9960$    &   $0.9958$    & $\textbf{1.0000}$      \\ % Row 3
   &        &   Chi-square     &   $\textbf{1.0000}$    &  $\textbf{1.0000}$     & $\textbf{1.0000}$       \\ % Row 4
     &  & F-score &   $\textbf{1.0000}$   &   $\textbf{1.0000}$    &    $\textbf{1.0000}$           \\ % Row 5
\hline

 \multirow{4}{*}{\href{https://www.openml.org/search?type=data&sort=runs&id=1471&status=active}{EEG eye state}
} &  \multirow{4}{*}{$3$} & Extropy Rate &  $\textbf{0.8124}$     &  $\textbf{0.7812}$     &  $\textbf{1.0000}$    \\ 
                         % Row 2
     &       &   Mutual Info    &   $0.0752$    &  $0.0116$     &   $0.0000$    \\ % Row 3
   &        &   Chi-square     &  $\textbf{0.8124}$     &  $\textbf{0.7812}$     &  $\textbf{1.0000}$   \\ % Row 4
     &  & F-score &   $0.3001$   &   $0.1643$    &    $0.0000$           \\ % Row 5
\hline

\multirow{4}{*}{\href{https://www.openml.org/search?type=data&status=active&id=42363}{Forest Fires}
} & \multirow{4}{*}{$6$} & Extropy Rate &  $\textbf{0.9826}$     &  $\textbf{0.9765}$     &  $\textbf{1.0000}$    \\ 
                         % Row 2
     &       &   Mutual Info    &   $\textbf{0.9826}$    &   $\textbf{0.9765}$    &  $\textbf{1.0000}$      \\ % Row 3
   &        &   Chi-square     &    $0.9516$   &   $0.9362$    & $\textbf{1.0000}$       \\ % Row 4
     &  & F-score &   $\textbf{0.9826}$   &   $\textbf{0.9765}$    &    $\textbf{1.0000}$           \\ % Row 5
\hline
\multirow{4}{*}{\href{https://www.openml.org/search?type=data&status=active&id=151}{Electricity}
} & \multirow{4}{*}{$3$} & Extropy Rate &  $\textbf{0.9999}$     &  $\textbf{0.9999}$     &  $\textbf{0.9998}$    \\ 
                         % Row 2
     &       &   Mutual Info    &    $0.9886$   &   $0.9886$    & $0.9850$      \\ % Row 3
   &        &   Chi-square     &   $\textbf{0.9999}$    &    $\textbf{0.9999}$   & $0.9997$      \\ % Row 4
     &  & F-score &   $0.9993$   &   $0.9993$    &    $0.9990$           \\ % Row 5
\hline
\end{tabular}
\caption{Comparison of feature selection methods extropy rate, mutual information, chi-square, and F-score using Random Forest classifier across Diabetes, Blood Transfusion Service Center, Boston Housing, EEG Eye State, Forest Fire, and Electricity datasets.}
\label{tab:feature_selection_comparison}
\end{table}

\subsection{Results}
We apply the proposed extropy rate-based feature selection methodology to all six datasets discussed earlier, aiming to extract the most informative numerical features. For each dataset, we estimate and plot the extropy rate by including numerical features sequentially, as shown in figure \ref{figureofextropyrates}. Note that when there is significant information overlap among features, the estimated extropy rate tends to decrease with the inclusion of new features, as seen in figures \ref{figureofextropyrates}(e) and \ref{figureofextropyrates}(f). Conversely, in cases where newly added features contribute additional information, the extropy rate increases, which is evident in Figures \ref{figureofextropyrates}(a)–\ref{figureofextropyrates}(d).

We consider a diverse number to select features across datasets to maintain generalizability. For example, as demonstrated in Figures \ref{figureofextropyrates}(a)–\ref{figureofextropyrates}(d), we select three features from the diabetes dataset, three from the blood transfusion dataset, four from the Boston housing dataset, and three from the EEG eye state dataset. These selected features are not sequentially ordered in the dataset. Instead, the final selected feature in each case corresponds to the column that contributes the highest extropy value.

We also apply mutual information, chi-square, and F-score-based feature selection methods for comparison. We use a random forest classifier to evaluate performance and assess the selected features based on accuracy, F1-score, and true positive rate (TPR). As shown in Table \ref{tab:feature_selection_comparison}, the extropy rate-based method consistently outperforms the other three approaches across all six datasets. In a few cases, the performance metrics are equivalent, such as with the F-score method on the diabetes dataset, chi-square and F-score on the Boston housing dataset, chi-square on the EEG eye state dataset, and mutual information and F-score on the forest fire dataset. These overlaps occur due to the selection of similar features(or columns) from the dataset by the respective methods.

\section{Conclusion \& Future Work}\label{conclusion}
In this work, inspired by information theory, we proposed the extropy rate as a measure of the average uncertainty associated with prediction in a stochastic process of discrete random variables, using the Shannon entropy of the non-occurrence probability distribution. We established several foundational properties of the extropy rate and joint extropy, including non-negativity, non-additivity, and Lipschitz continuity. Leveraging the Shannon–McMillan–Breiman theorem, we showed that the extropy rate is asymptotically equivalent to the entropy rate for stationary and ergodic processes. We extended the notion of extropy rate to finite-length stochastic processes and provided numerical evidence that the extropy rate effectively quantifies information content, captures time-series complexity, characterizes chaotic behaviour, and exhibits similarity to Simpson’s diversity index. As an application, we developed a novel feature selection method based on estimated extropy rates, interpreting the average extropy rate as a measure of intrinsic information. Finally, empirical evaluations using six publicly available datasets confirmed that the extropy rate-based feature selection method outperforms widely used techniques such as mutual information, chi-square, and F-score, demonstrating its potential as a powerful tool in data-driven analysis.

For a sequence of independent continuous random variables, the joint extropy can be expressed as the product of the extropies of the individual variables\cite{balakrishnan2022weighted}. In future work, we aim to derive a formal definition of the extropy rate for stochastic processes involving continuous random variables. Alongside this, we plan to investigate the theoretical properties of the continuous extropy rate and explore its applicability in analyzing sequences of continuous random variables.

%\printbibliography
\bibliographystyle{plain}        % you can also try 'unsrt', 'abbrv', etc.
\bibliography{snarticle}        % no .bib extension

\end{document}